\documentclass[journal,draftcls, onecolumn, 12pt]{IEEEtran}
\pdfoutput=1

\usepackage[usenames,dvipsnames]{xcolor}

\usepackage[tight,footnotesize]{subfigure}
\usepackage[pdftex]{graphicx}

\usepackage{graphics} 
\usepackage{graphicx}
\usepackage{verbatim}
\usepackage{color}
\usepackage{subfigure}

\usepackage{caption}
\usepackage{amsmath} 
\usepackage{amssymb}  

\usepackage{mathtools}

\usepackage{amsmath}
\makeatletter
\newcommand{\distas}[1]{\mathbin{\overset{#1}{\kern\z@\sim}}}%
\newsavebox{\mybox}\newsavebox{\mysim}
\newcommand{\distras}[1]{%
	\savebox{\mybox}{\hbox{\kern3pt$\scriptstyle#1$\kern3pt}}%
	\savebox{\mysim}{\hbox{$\sim$}}%
	\mathbin{\overset{#1}{\kern\z@\resizebox{\wd\mybox}{\ht\mysim}{$\sim$}}}%
}

\newtheorem{theorem}{Theorem}

\newtheorem{corollary}{Corollary}
\newtheorem{lemma}{Lemma}

\def\conv{\otimes}
\def\deconv{\oslash}
\def\minim{\wedge}

\def\eps{\varepsilon}
\def\P{{ Pr}}  
 
\def\S{{\cal S}}

\def\F{{\mathcal{F}}}

\def\Swin{{S_{\rm win}}}

\newcommand{\add}[1]{\textcolor{magenta}{#1}}

\begin{document}

\title{Window Flow Control Systems with Random Service} 
\author{Alireza Shekaramiz*, 
         J\"{o}rg Liebeherr*, Almut Burchard**  \\[20pt] 
        * Department of ECE,   University of Toronto,  Canada. \\[-10pt]
        ** Department of Mathematics,   University of Toronto, Canada. \\[-5pt]
	 	E-mail:  \{ashekaramiz, jorg\}@ece.utoronto.ca; almut@math.toronto.edu.
        }%

\setcounter{page}{1}
\maketitle
\thispagestyle{plain}
\pagestyle{plain}

\begin{abstract}
We present an extension of the window flow control analysis  by R. Agrawal et. al. (Reference \cite{Cruz99}), C.-S. Chang (Reference \cite{Chang98a}), and C.-S. Chang et. al. (Reference \cite{ChangCBT02})  
to a system with random service time and fixed feedback delay. 
We consider two network service models. In the first model, the 
network service process itself has no time correlations. The second model 
addresses a two-state Markov-modulated service.  

\end{abstract}

\section{Introduction}

The discovery that deterministic feedback systems can be 
expressed in the 
network calculus  as the solution of a fixed-point equation in the dioid algebra 
was exploited for an analysis of a general feedback system  \cite{Cruz99,Chang98a,ChangCBT02}. 
The analysis was conducted for deterministic systems, 
which cannot exploit statistical multiplexing of traffic flows.  
A detailed feedback model of the TCP Tahoe and TCP Reno algorithms as 
max-plus linear systems was presented in \cite{baccelliTCP}. The analysis accounts for 
randomness in the feedback system, but assumes that the underlying service process is deterministic.  A further simplification is that the derivations in   \cite{baccelliTCP} are conducted for an overloaded system, with a claim that the results extend to any load condition. 
Overcoming the limitations of a deterministic analysis has motivated the development of the {\it stochastic network calculus}, where traffic and service are characterized by random processes \cite{Book-Jiang}. Whereas, over the last ten years,  many key results of the network calculus have been shown to also hold under probabilistic assumptions, there has not been a successful treatment of stochastic feedback systems, see \cite[p. 82]{Fidler2010}, \cite[p. 104]{Fidler2015}, \cite[pp. 215-216]{Book-Jiang}.

\begin{figure}
\centering
\includegraphics[width=0.6\columnwidth]{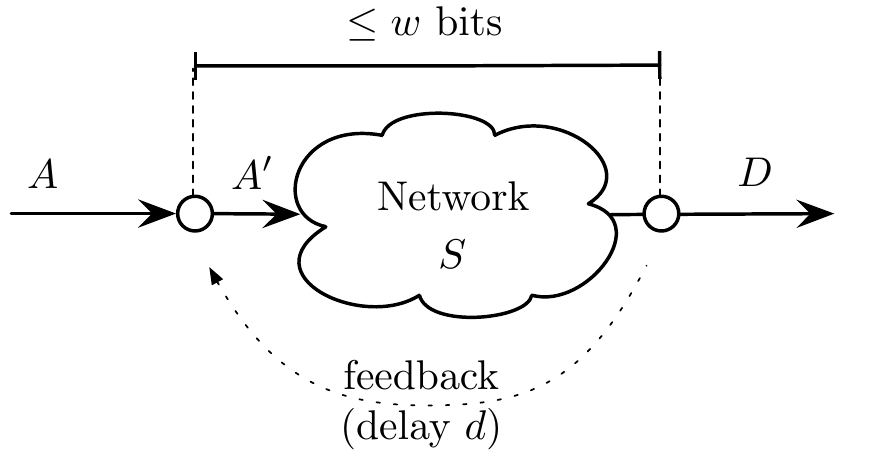}
\caption{Network with window flow control. Traffic is admitted to the network only if its backlog does not exceed $w$ bits. Feedback, with delay $d$, informs the network ingress about departures. Random processes $A, A', D$, and $S$, respectively, describe the external arrivals, the admitted arrivals, the departures, and the network service.} 
\label{chp2-fig:window-flow-control}
\end{figure}
In this paper, we present an  analysis of stochastic feedback systems 
with methods of the stochastic network calculus. 
Our paper extends the deterministic network calculus analysis of  feedback systems in \cite{Cruz99,Chang98a,ChangCBT02,Book-LeBoudec}.
Following prior works, we study a window flow control system as shown in Fig.~\ref{chp2-fig:window-flow-control}, which enforces that only a limited number of 
$w$ bits, referred to as  {\it window size},  can be in transit at any time.  The feedback consists of (acknowledgement) messages, which relay  information on the amount of departing traffic back to  the network entrance. 
The delay of the feedback is assumed to be $d$ time units. 
The feedback signal opens or closes a throttle that prevents traffic from  entering the network.

Selecting a discrete-time model for the analysis, we  derive upper and lower 
bounds on the service experienced by a 
traffic flow where the external arrivals 
and the available network service are characterized by 
bivariate random processes. 
We are interested in gaining insight into the service impediment caused by the 
feedback mechanism  with given window size $w$ and feedback delay $d$. 
This is done by deriving an equivalent service process, which describes 
the available service when taking into account the window flow control constraints.

Our analysis applies the {\it moment-generating function (MGF) network 
calculus}  from \cite{Fidler06}.
One challenge for the analysis of a stochastic feedback system with the MGF network 
calculus is that the standard method to compute the frequently encountered  
convolution operation is not well-suited for feedback systems. 
We adapt the MGF network calculus to feedback systems by deriving a sharpened version of 
the  convolution estimate from \cite{Fidler06}. 


It remains  open to which degree our results can be applied to other feedback systems. 
Even though our paper only presents single node results, we emphasize that available
techniques  of the MGF network calculus permit an immediate 
extension to a multi-node network consisting of a sequence of feedback 
systems. More complex network topologies, in particular, nested feedback systems 
are not covered.

The rest of the paper is structured as follows. In Sec.~\ref{sec:Preliminaries}, we provide background on network calculus, analysis of feedback systems, and the MGF network calculus with bivariate processes. 
In Sec.~\ref{sec:stochfeedback}, we discuss the obstacles to a network calculus analysis of 
stochastic feedback systems, and present the main results of this paper. In Sec.~\ref{sec:cbrvbr}, 
we consider a feedback system with an i.i.d. service model. 
Sec.~\ref{sec:MMOO} extends the analysis to a service model with time correlations. 
We present brief conclusions in Sec.~\ref{sec:conclusions}.

\section{Bivariate Network Calculus} \label{sec:Preliminaries}

The derivations in this paper will be done in the context of a bivariate network calculus, 
which has been used for  deterministic as well as stochastic analyses 
of networks \cite{Book-Chang}. 
We consider a discrete-time domain with $t = 0, 1, 2, \ldots$ describing 
time slots.

The vast majority of network calculus research is done with univariate 
functions $f(t)$ that quantify events in a time interval $[0,t)$. 
A bivariate network calculus uses functions $f(s,t)$ to characterize events in 
the time interval $[s,t)$.  The rationale for preferring the univariate network calculus is 
that it has stronger algebraic properties. 
On the other hand, univariate functions do not lend themselves easily to probabilistic extensions. 
For example, consider a deterministic system that offers a time-invariant  
service $S(s,t)$ for each time interval $[s,t)$. Since  $S(s,t)= S(0,t-s)$,  
the service can be completely described by a univariate function  $S(\tau)$ expressing the service in any time interval of length $\tau$. 
In the corresponding probabilistic model, if the service is stationary, $S(s,t)$ and $S(0,t-s)$ are equal 
in distribution, however,  generally $S(s,t)\neq S(0,t-s)$, thus prohibiting the 
convenient   reduction to  univariate functions.

\subsection{$(\minim, \conv)$-dioid algebra for bivariate functions}
We denote by $\F$ the family of bivariate functions that are non-negative and non-decreasing in the second argument. 
We use ${\F}_o$ to denote the set of functions $f \in \F$ with 
$f (t,t)=0$, which we refer to as causal functions. 
For $f, g \in  \F$,  the minimum ($\minim$) and 
convolution  ($\conv$)  operations are defined by 
\begin{align*}
&  f\minim g \, (s,t)=\min\{f(s,t),g(s,t)\} \, , \\
& f\otimes g \, (s,t)=\min_{s\leq\tau\leq t}\{f(s,\tau)+g(\tau ,t)\} \, . 
\end{align*}

The two operations form a dioid algebra on the sets $\F$ and $\F_o$ \cite{Book-Chang}. 
The convolution operation in these dioids is associative, but not commutative. 
(In contrast, the convolution in the corresponding dioid algebra for 
univariate functions is commutative.)  
Also, the convolution distributes over the minimum, 
i.e., for three functions $f$, $g$, and $h$ we have
 $f \conv (g \minim h) = f \conv g\minim f \conv h$.\footnote{Statements that list 
functions without arguments hold for all pairs $(s,t)$ with $s \le t$. 
Also, we give the  
 convolution operation  precedence over the  minimum, which allows us to omit some parentheses.} The function  
\[
\delta (s, t) = \begin{cases}
0\, & s \ge t \, , \\  
\infty  \, & s <  t \, ,
\end{cases} \
\]
is the neutral element of the convolution operation, with $f \conv \delta = \delta \conv f = f$ for any $f \in \F$. 

For the analysis of feedback systems, we need to convolve a function multiple times with itself.
For that, we use the notation  
\begin{align*}
& f^{(0)}  =\delta \, , \quad f^{(1)}  = f \, , \\
& f^{(n+1)}  = f^{(n)} \conv f \, , \qquad (n\ge 1) \, . 
\end{align*} 
We  define the subadditive closure of $f \in \F$, denoted by $f^{*}$, as
\begin{align*}
		f^{*}  &=\delta \minim f\minim f^{(2)}\minim f^{(3)}\minim \ldots  
		=\bigwedge_{n=0}^{\infty}f^{(n)} \, .  
	\end{align*}
The attribute  `subadditive' derives from the property that  $f^* (s,t)\le f^* (s, \tau)+ f^* (\tau, t) $ for $s \le \tau \leq t$. 
Every subadditive function $f$ satisfies $ f \conv f =f$. 
If $f \in \F_o$, then its subadditive closure is given by $f^{*} = \lim_{n\to\infty}f^{(n)} $. 

\subsection{Bivariate arrival and service processes}
\label{subsec:DetNetCalc}
The arrivals to a network element are characterized by a bivariate process $A$, 
where $A(s,t)$ describes the  
cumulative arrivals in the time interval $[s,t)$. 
We use $D(s,t)$ to describe the  
 departures in the time interval $[s,t)$, 
 subject to the  causality condition $D (0,t) \le A (0,t)$.  
Both $A$ and $D$ are causal functions ($A, D \in \F_o$).
If $a_k$ and $d_k$ denote the arrivals and departures, respectively, 
at time $k$, we have 
\[
A(s,t) = \sum_{k=s}^{t-1} a_k  \quad \text{and} \quad 
D(s,t) = \sum_{k=s}^{t-1} d_k \, . 
\] 

The available service at a  network element is described by a bivariate 
service process $S(s,t) \in \F_o$ that 
satisfies the input-output relationship 
\begin{align*} 
D \geq A \conv S       
\end{align*} 
for any arrival process and corresponding departure process at that element. 
Such a process is called 
a {\it dynamic server} in \cite[p. 178]{Book-Chang}. 
An {\it exact} service process satisfies the input-output relationship with equality. 
The service offered by a sequence of network elements with service processes $S_1, S_2,\ldots,S_N$ is given by 
their  convolution  $S_1 \conv S_2 \conv \ldots \conv S_N$.

Given arrival and service processes, one can formulate bounds on backlog  at a network element. The backlog $B$ of a network element denotes the arrivals that have not yet departed, given by $B(t)=A(0,t)-D(0,t)$.  A bound on the backlog can be expressed in terms of the deconvolution operation ($\deconv$), 
which, for two functions $f, g \in  \F$,  is defined by  
\begin{align*}
f\deconv g(s,t)=\max_{0\leq\tau\leq s}\{f(\tau,t)-g(\tau ,s)\} \, . 
\end{align*}		
With this operation we have \cite[Theorem~2]{Fidler06} 
\begin{align}\label{eq:backlogbound}
B(t) \le A\deconv S(t,t) \, .   
\end{align}
Bounds on the delay and the burstiness of departures also involve the deconvolution operation.

\subsection{Bivariate feedback systems} \label{subsec:FeedbackSystem} 
	
		\begin{figure}
				\centering
				\includegraphics[scale=1]{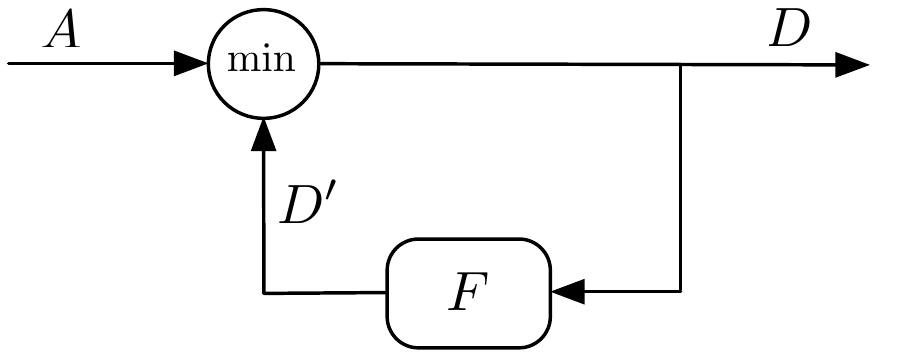}
				\caption{ Generic feedback server.}
				
				\label{fig:FeedbackServer}
		\end{figure}

A network element with feedback is one where the departures from the element influence the arrivals to the same element. Fig.~\ref{fig:FeedbackServer} depicts a generic system-theoretic model of a closed-loop feedback system.  There is a network element with service process $F$ whose output is re-combined 
with the external input, so that the arrivals to the network element 
are the minimum of its output and the external arrivals. Such a system will be 
referred to as {\it feedback server}. 
The feedback server consists of a (non-causal) service element with service process $F(s,t)$ ($F \not\in \F_o$) and 
arrival and departure functions $D' \ge D \conv F$. The service element labeled by `min' in Fig.~\ref{fig:FeedbackServer} represents a 
{\it throttle}, which enforces the minimum $D = A \minim D'$. With this, $A$, $D$, and $F$ satisfy  
$D \ge  A \minim D \conv F$. 
As shown in \cite{Cruz99,Chang98a} for  univariate functions, the closed-loop system can be replaced by an equivalent system, which consists of 
a single network element without feedback with service process $F^*$, where 
$F^*$ is the subadditive closure of $F$. 
In  \cite{ChangCBT02}, the result has been extended to  bivariate functions, 
as expressed in the following lemma. 		
		\begin{lemma}{\it (see \cite[Lemma 2.2]{ChangCBT02})} 
		\label{lemma:Chang-feedback}
			 Given a feedback server as in Fig.~\ref{fig:FeedbackServer} with bivariate functions $A,D$, and $F$. If $F \not\in \F_o$ and 
$D\geq A \minim D\otimes F$, then $D\geq A\conv F^{*}$. 

		\end{lemma}	

Since the lemma applies when  
$A, D$, and $F$ are random processes, an analysis of stochastic feedback system appears 
readily available. However, computing the subadditive closure $F^{*}$ for a 
bivariate random service process presents considerable difficulties.   

\subsection{MGF network calculus}

The MGF network calculus  \cite{Fidler06} offers an analysis 
of network elements, when the arrivals and the offered service are 
bivariate random processes that are characterized in terms 
of their moment-generating functions. The MGF network calculus has been frequently applied to the analysis of wireless networks, 
e.g., \cite{alzubaidy,Fidler-Fading,Mahmood_Rizk_Jiang,ZhengTWC2012}, since the random service model can capture the randomness of a wireless transmission system. 
We denote the moment-generating function of a random variable $X$ for any 
$\theta \in \mathbb{R}$ by $M_{X}(\theta)=E[e^{\theta X}]$. 
The MGF calculus exploits that for any two independent random variables $X$ and $Y$, the relation $M_{X+Y}(\theta) = M_{X}(\theta) M_{Y}(\theta)$ holds. 
The stochastic analysis of feedback systems in this paper will 
take an MGF network calculus approach.  

The moment-generating functions of an arrival process $A$ and a 
service process $S$ for $\theta >0$ are denoted by  
\[
M_A (\theta, s,t) = E[e^{\theta A(s,t)}] \quad \text{and}\quad
 M_S (-\theta, s,t) = E[e^{-\theta S(s,t)}] \, . 
\]
We assume that the arrival and the service are independent.  
For characterizing the random service of a feedback system, we will use 
a bivariate version of the statistical service curve from \cite{BuLiPa06} and the effective capacity from \cite{WuNegi03}. 
A statistical service curve $\S^\eps$ of a bivariate service process $S$ for a given $\eps >0$ \cite{Fidler06} 
is defined by  the property that 
\begin{align}
\P{\Bigl(S  (s,t) \leq {\cal S}^{\eps} (s,t)\Bigl)}\leq \eps \, . 
\label{eq:Srvc-probbound-0}
\end{align}
The statistical service curve is a deterministic function giving a lower bound 
on the available service that is violated with a probability  $\eps$ or less. 
Using the Chernoff bound, a statistical service curve can be computed from $M_S$  as  
\begin{align}
{\cal S}^{\eps}  (s,t) =  \max_{\theta >0}\frac{1}{\theta}
\Bigl\{\log{ \eps}-\log{M_S (-\theta, s,t)}\Bigl\} \, .  
\label{eq:Srvc-probbound}
\end{align}
An alternative measure to describe a stochastic service in terms of  $M_S$ is the 
effective capacity $\gamma_S (-\theta)$ \cite{WuNegi03},  defined for  $\theta>0$ by 
\begin{align}
\gamma_S(-\theta) = \lim_{t\to\infty} - \frac{1}{\theta t} \log M_S (-\theta,0,t) \, . 
\label{eq:Eff-capacity}
\end{align}
Since it is defined as a time limit, the effective capacity is most useful when 
reasoning about long-term traffic rates and scaling properties. In particular, 
$\gamma_S(0) = \lim_{\theta \to 0}\gamma_S(-\theta)$ 
equals the average service rate.
In contrast, the statistical service curve provides a bound for finite values of $t$. 

Given the moment-generating functions of two bivariate processes $f$ and $g$,  bounds on the convolution and deconvolution, given in \cite{Fidler06}, are  
\begin{align}
M_{f \conv g} (-\theta, s,t)  & \le \sum_{\tau = s}^t M_{f} (-\theta, s,\tau)M_{g} (-\theta, \tau,t) \, , \label{eq:mgf-conv} \\
M_{f \deconv g}  (\theta, s, t) & \le 
\sum_{\tau=0}^{s}  M_{f}(\theta, \tau, t) M_{g}(-\theta, \tau, s) \, . 
\label{eq:mgf-deconv}
\end{align}
The convolution bound is used to estimate the moment-generating function of a sequence of  service elements. 
The deconvolution bound 
plays a role when computing performance bounds. For example, with an application of  the Chernoff bound, we can obtain from Eq.~\eqref{eq:backlogbound} a  
bound  on the backlog distribution $\P ( B(t) > b^*(t) ) \le \eps$ \cite{Fidler06}, where 
\begin{align}\label{eq:backlog-MGF}
b^*(t) = \min_{\theta >0 } \frac{1}{\theta} \left\{ \log \left( \sum_{\tau=0}^{s} M_{A}(\theta, \tau, t) M_{S}(-\theta, \tau, s) \right) - \log \eps \right\} \ . 
\end{align}
The delay can be treated in a similar fashion.

\section{Towards a Stochastic Feedback Analysis}\label{sec:stochfeedback}

The generalization of an analysis of a deterministic feedback system to 
random arrival and service processes has remained open   
for considerable time. In this section, we present results that make such an 
analysis possible. We describe the issues 
that make a stochastic analysis of feedback systems within the framework 
of the network calculus hard, and then address how to resolve them.

\subsection{Model Description} \label{subsec:model}

We will analyze the network with window flow control in Fig.~\ref{chp2-fig:window-flow-control}. 
Traffic with  arrival process $A(s,t)$ is serviced by a network 
element with 
service process $S(s,t)$, subject to the additional constraint 
that the total backlog in the element at any time 
may not exceed $w >0$.   
Traffic in excess of that constraint is 
held in a FIFO buffer at the network entrance. The traffic leaving the network 
is expressed 
by a departure function $D (s,t)$. 
The feedback information consists of the value of the departure function 
delayed by $d \ge 0$ time slots, i.e., $D (s,t-d)$. 
If we use $A'$ to denote the arrived traffic that is admitted into the network,  the  flow control system 
requires that   
\begin{align} \label{eq:Aprime}
A'(s,t) = \min \left\{ A(s,t)  , D(s,t-d) +w \right\}  \, .  
\end{align}
This control ensures that admitted traffic 
that has not yet departed cannot exceed $w$.

The system can be described by a feedback model as discussed in Subsec.~\ref{subsec:FeedbackSystem}. 
Let the network service be given by a service process $S$, i.e., 
\begin{align}
\label{eq:conv-A'}
D \geq A' \conv S \, . 
\end{align}
Further, for  $w > 0$ we define a function $\delta^{+w}$  by  
\begin{align}
\label{eq:delta-w-def}
\delta^{+w} (s,t) = \begin{cases}
w \, & s \geq t \, , \\
\infty \, & s < t \, .   
\end{cases} 
\end{align}
This allows us to write   $f(s,t) + w =  f \conv \delta^{+w} (s,t)$. 
Note that $\delta^{+w}$ is not a causal function, i.e., $\delta^{+w}> 0$ for 
$t \le s$. 
Even though the convolution 
of bivariate functions is generally not commutative, we have  $f \conv \delta^{+w} =  \delta^{+w} \conv f$
 for a bivariate process $f$.
Lastly, 
a  service element offering a delay of $d\ge 0$ is 
given by a service process $\delta_d = \delta (s,t-d)$, 
such that $f(s,t-d) = f\conv \delta_d (s,t)$ for every $f \in \F$. 
	\begin{figure}
		\centering
		\includegraphics[scale=1]{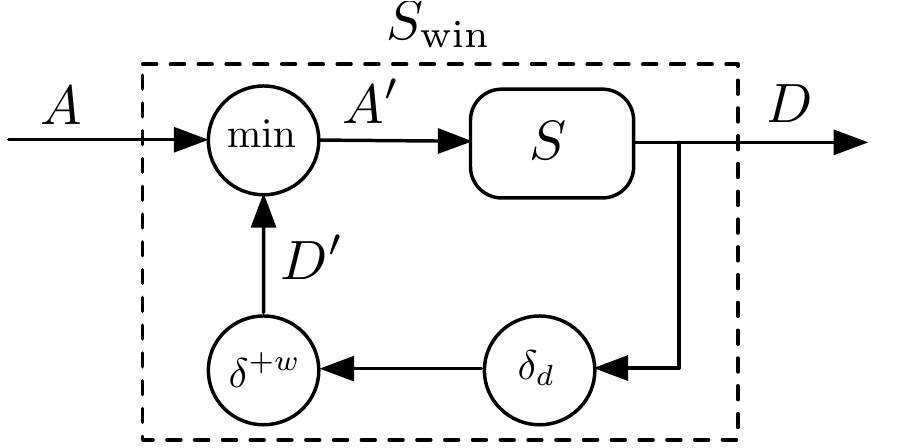}
		\caption{Model with flow control with window size $w >0$ 
and feedback delay $d \ge 0$.} 
		\label{fig:feedback-w-d}
	\end{figure}
With these definitions we rewrite Eq.~\eqref{eq:Aprime} as 
\begin{align} \label{eq:Aprime2}
A'  =   A  \minim  D  \conv \delta_d \conv \delta^{+w}  \,  .   
\end{align}
In Fig.~\ref{fig:feedback-w-d}, we illustrate the feedback system as a concatenation of service elements. 
By inserting Eq.~\eqref{eq:conv-A'} into Eq.~\eqref{eq:Aprime2} we obtain 
\begin{align} \label{eq:Aprime3}
A'  \ge   A  \minim  A' \conv S  \conv \delta_d \conv \delta^{+w}  \,  .   
\end{align}
Note that $A'$ satisfies the conditions of   Lemma~\ref{lemma:Chang-feedback} with $F = S  \conv \delta_d \conv \delta^{+w}$. Applying the lemma and inserting the result into Eq.~\eqref{eq:conv-A'} yields  
\[
D \geq A   \conv (S \conv  \delta_d \conv \delta^{+w})^*  \conv S\ .  
\]
Therefore, the  flow control system can be represented by an equivalent 
service process $S_{\rm win}$ given by  
\begin{align}
\label{eq:swin}
S_{\rm win}=  (S \conv  \delta_d \conv \delta^{+w})^* \conv S \ . 
\end{align}
If $S$ is an exact service process, so is $S_{\rm win}$.


\subsection{Challenges of stochastic feedback systems}

The analysis of a feedback system with a bivariate random service process $S$ 
requires the computation of the subadditive closure in Eq.~\eqref{eq:swin}.  
Let us consider for the moment a system with no feedback delay, that is, $d = 0$. 
Since $(S \conv  \delta^{+w}) \conv (S \conv  \delta^{+w}) = S \conv  S \conv \delta^{+2w}$, 
the computation of the subadditive closure involves the convolution of $S$ with itself. 
When the service process is subadditive, we have 
$S = S \conv S$, and therefore obtain 
$(S \conv  \delta^{+w})^* \conv S = S$. 
Hence, when a  subadditive process $S$ is the 
 service process of a closed-loop system without feedback delay, feedback has no impact on the overall service.   
 On the other hand, a service process with a nonzero feedback delay is 
generally not subadditive and, therefore, the convolution of such a process will not yield a trivial result. 
To see that the concatenation of $S$ and $\delta_{d}$ with $d >0$ is not subadditive,  we compute 
\begin{align*}
S\conv \delta_d (s,\tau)  +  S\conv \delta_d (\tau,t )  
& = S ( s , \tau - d) +  S ( \tau , t - d)  \\
& < S ( s , t - d) \\
& = S\conv \delta_d (s,t) \ . 
\end{align*}
To avoid trivial cases, we will henceforth consider feedback systems with $d >0$.

Writing the expression for the equivalent service process 
in Eq.~\eqref{eq:swin} as 
\begin{align}
\begin{split}
\label{eq:swin-w-d}
S_{\rm win} 
	&=\bigwedge_{n=0}^{\infty}\biggl((S\conv \delta_d\conv \delta^{+w} )^{(n)}\otimes S\biggl) \,   
\end{split}
\end{align} 
makes apparent the need for tools to analyze the distribution
of minima involving a bivariate random service process $S$.  
Note that the $n$-fold convolutions themselves
involve minima as well as sums of random processes.

Consider, for comparison, the same formula
in the case of an univariate service process $S(t)$. 
There, one can
exploit the commutativity of the convolution to obtain
for the $n$-th term in the minimum
$$
(S\conv \delta_d \conv \delta^{+w})^{(n)}\conv S (t)
= S^{(n+1)}\conv \delta^{+nw}\conv \delta_{nd}(t)
= S^{(n+1)}(t-nd) + nw\,.
$$
In particular, if $S$ itself is subadditive, then $S^{(n)}=S$,
the $n$-th term is $S(t-nd) + nw$,
and the entire minimum
reduces to $S_{\rm win}=S\conv S_o$, where 
$S_o(t)=w\left\lceil\frac td \right\rceil$. Since deterministic feedback systems 
can generally be expressed using univariate service processes, the computation of 
the subadditive closure of the service process is much simplified. 
However, for bivariate service processes, the convolution is not commutative, 
and 
we must compute the $n$-th term as
\begin{align}
\label{eq:swin-n}
(S\conv \delta_d \conv \delta^{+w})^{(n)}\conv S (s,t)
&=\min_{\tau_o\le \dots \le \tau_{n}}
\left\{
\sum_{i=1}^n
\Bigl( S(\tau_{i-1},\tau_i-d)\Bigr) + S(\tau_{n},t)
\right\} + nw\,,
\end{align}
where the minimum ranges over all non-decreasing
sequences $\tau_o,\dots,\tau_n$ with $\tau_o=s$ and $\tau_n\le t$.
One issue with this expression is that the number of terms grows rapidly with $n$ and 
$t-s$. The second difficulty is that the 
joint distribution of $S(\tau_{i-1},\tau_i-d)$ for $i=1, \ldots, n$ is not determined by the distribution of $S(s,t)$
alone, but requires information
on time correlations.
We now proceed to address these problems.

\subsection{Service processes of stochastic feedback systems}

Our first result provides an exact characterization of $\Swin$ for the special case $d=1$.

\bigskip
\begin{lemma} 
\label{lem:Swin-d1-exact}
Consider a feedback system as in Fig.~\ref{fig:feedback-w-d} 
with an additive service process 
\begin{align*}
S(s,t)& =\sum_{k=s}^{t-1} c_k\,,
\end{align*}
where $(c_k)_{k\ge 1}$
is an arbitrary sequence of non-negative random variables.
For $d=1$ and $w>0$, the equivalent service process
is given by
\begin{align*}
S_{\rm win}(s,t) = \sum_{k=s}^{t-1} \min \{c_k, w\}\,.
\end{align*}
\end{lemma}

\bigskip\noindent 
\begin{proof} Set $d=1$. For $n=1$, the convolution
in Eq.~\eqref{eq:swin-n} is given
by \begin{align*}
S\conv \delta_1\conv\delta^{+w}\conv S(s,t) 
&=\min_{s\le \tau\le t} 
\bigl\{ S(s,\tau-1) + S(\tau,t)\bigr\} + w\\
& = \min_{s\le \tau\le t} \biggl\{\sum_{k=s}^{\tau-2} c_k +\sum_{k=\tau}^{t-1}
c_k
\biggr\}+ w\\
&=S(s,t) -\max_{k\in [s,t)} c_k + w\,.
\end{align*}
Here, the convolution replaces the largest
value in the sum representing $S(s,t)$
by $w$. Likewise, for $n>1$, the minimum in
Eq.~\eqref{eq:swin-n} is obtained by replacing the $n$ 
largest values of $c_k$ on $[s,t)$
with $w$. In this case, the minimum in Eq.~\eqref{eq:swin-w-d} is attained for 
 $n=\#\{k\in [s,t)\,\mid\,c_k>w\}$.
\end{proof}

\bigskip

The following theorem shows that the
exact result from Lemma~\ref{lem:Swin-d1-exact}
provides a lower bound on $S_{\rm win}$
for $d>1$ in terms of the ratio $w/d$. The theorem also provides a complementary upper bound on $S_{\rm win}$.

\begin{theorem}
\label{thm:Swin-apriori}
Given a  service process
$S \in \F_o$. Let $S_{\rm win}$ be the equivalent service process
for a feedback system with parameters $d>0$ and $w>0$.
Then
\begin{equation}
\label{eq:Swin-apriori}
S'(s,t)\ \le \ S_{\rm win}(s,t)
\ \le \ \min\left \{S(s,t),\left\lceil\tfrac{t-s}{d}\right\rceil w
\right\} 
\,,\quad 0\le s\le t\,,
\end{equation}
where $S'$ is the equivalent service process (given by Eq.~\eqref{eq:swin-w-d}) 
with $d'=1$ and $w'=w/d$.
\end{theorem}

\bigskip

\begin{proof} For the lower bound, 
observe that \begin{align}
\delta_d \conv \delta^{+w} = (\delta_1 \conv \delta^{+w'})^{(d)}\,.
\end{align}
This implies
\begin{align*}
S\conv \delta_d\conv\delta^{+w} (s,t)
&\ge \min_{s=\tau_o\le \dots\le \tau_d=t}
\biggl\{ \sum_{i=1}^d \Bigl( S(\tau_{i-1},\tau_i-1)+w'\Bigr)
\biggr\} \\
&= \left(S\conv \delta_1\conv \delta^{+w'}\right)^{(d)} (s,t)\,.
\end{align*}
Indeed, the left-hand side appears as a term in 
the minimum on the right-hand
side. Therefore, the $n$-th term in the expression
for $S_{\rm win}$ in Eq.~\eqref{eq:swin-w-d} is bounded
from below by the $(nd)$-th term in the expression
for $S'$. 

For the upper bound, we simply use
just two terms from the minimum in Eq.~\eqref{eq:swin-w-d}
to bound $S_{\rm win}$, 
namely $n=0$ and $n= \left\lceil\tfrac{t-s}{d}\right\rceil$. 
\end{proof}

An implication of the theorem is that if $S$, $S_{\rm win}$, and $S_{\rm win}'$ 
have long-term  average rates $C$, $C_{\rm win}$, and $C_{\rm win}'$, then 
$$
C'_{\rm win}\ \le \ C_{\rm win} \ \le\ 
\min \left\{C, \frac w d \right\}\,.
$$

\bigskip 
Theorem~\ref{thm:Swin-apriori}
holds for general (deterministic or random) service processes.
The bounds on the average rates
depend only on the ratio $w/d$, not on
the values of the individual parameters.
Our main results, in Theorems~\ref{thm:Swin-VBR} and~\ref{thm:Swin-MMOO},
will strengthen the lower bound when $d>1$ for two 
important classes of additive service processes.
They rely on the next lemma,
which reduces the number of terms
that contribute to Eq.~\eqref{eq:swin-n} 
and limits the
range of the minimum in Eq.~\eqref{eq:swin-w-d}
to $n\le \left \lceil\frac{t-s}{d}
\right\rceil$. 

\begin{figure}[t!]
	\centering
	\includegraphics[scale=0.5]{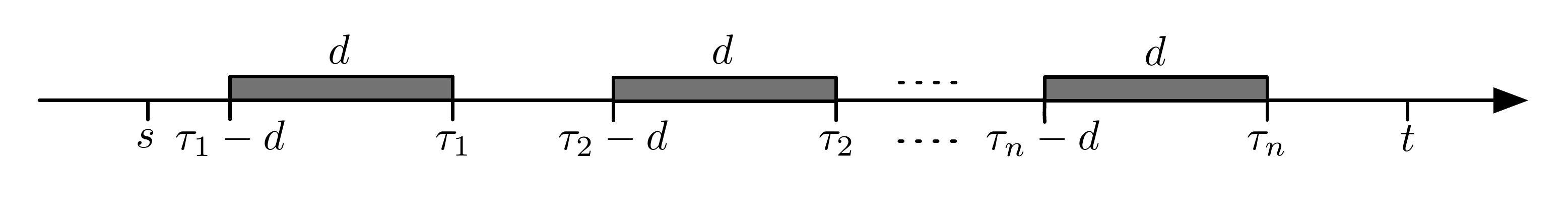}
	\caption{Geometric interpretation of the index set $C_n(s,t)$.}
	\label{fig:convForndelayIntervals}  
\end{figure}
\bigskip

\begin{lemma} 
\label{lem:Swin-giganticMin}
Given a feedback system with a  service process
$S \in \F_o$.
Then, for any choice of 
$d>0$ and $w>0$,
\begin{align} 
\label{eq:Swin-giganticMin}
S_{\rm win}(s,t) =
\bigwedge_{n=0}^{\left\lceil \frac {t-s}{d}\right \rceil}
\left\{ \min_{\tau_o,\dots, \tau_n\in C_n(s,t)}
\left(
\sum_{i=1}^{n}
S(\tau_{i-1},\tau_i-d)+S(\tau_n , t)\right)+ nw\right\} \,,
\end{align}
where the minimum in the braces ranges over
\begin{align*}
 C_n(s,t) = \left\{\tau_o, \tau_1, \ldots, \tau_n \, \Big\vert \, 
\tau_o=s, \tau_n\le t,\ \text{and}\ 
\tau_i - \tau_{i-1} \ge d \ \text{for}
\ i=1,\dots,n\right\}\, , 
\end{align*} 
if $nd\le t-s$. If $nd>t-s$, then $C_n(s,t)$ contains the single
sequence $\tau_i =(s+id)\wedge t$ for $i=0,\dots, n$.

\end{lemma}

\bigskip
The set $C_n(s,t)$ has a 
geometric interpretation,  illustrated in Fig.~\ref{fig:convForndelayIntervals}.
Each sequence $\tau_o,\dots,\tau_n$ in $C_n(s,t)$
corresponds to a collection
of $n$ disjoint subintervals of length $d$ in
$[s,t]$, given by $(\tau_i-d,\tau_i]$. 
On each of these subintervals, the original 
service process $S(\tau_i-d,\tau_i]$ is
interrupted by a delay of length $d$, followed by
an addition of $w$.  In the special case $d=1$,
the set $C_n(s,t)$ consists
precisely of the $n$-element subsets of $[s,t)$,
and we recover the statement of Lemma~\ref{lem:Swin-d1-exact}.
 
\bigskip

\begin{proof} 
We need to prove that only sequences
$\tau_o,\dots,  \tau_n$ in $C_n(s,t)$ contribute
to the minimum in Eq.~\eqref{eq:swin-n}.  
If $nd> t-s$, the minimum value is achieved by subdividing
$[s,t)$ into $n$ subintervals of length at most $d$,
so that all terms involving $S$ vanish.
Then the right-hand side of Eq.~\eqref{eq:swin-n}
equals $nw\ge \left\lceil\frac{t-s}{d}\right\rceil w$.

For $nd\le t-s$, we proceed by induction over $n$.
The $n=0$ term in both Eq.~\eqref{eq:swin-w-d} and
Eq.~\eqref{eq:Swin-giganticMin} is given  by $S(s,t)$ and there
is nothing to show. The $n=1$ term
equals
\begin{align}
\label{eq:convFor2delayIntervals}
S\conv \delta_d\conv \delta^{+w}\conv S(s,t)
&= \min_{s\le\tau\le t} \bigl\{
S(s,\tau-d) + S(\tau, t)\bigr\} + w\,.
\end{align}
The key observation is that the expression in the braces is 
non-increasing in $\tau$ for $s\le\tau\le d$,
since the first term vanishes while
the second term is always non-increasing. 
Therefore, the minimum is achieved for some $\tau\ge s+d$.
This means that the sequence $\tau_o=s, \tau_1=\tau$ lies in $C_1(s,t)$.

Now consider $n\ge 1$. Suppose we already know that
\begin{align*}
(S\conv \delta_d \conv \delta^{+w})^{(n)}\conv S (s,t)
= \min_{\tau_o,\dots,\tau_n\in C_n(s,t)}
\sum_{i=0}^n
\Bigl(S(\tau_{i-1},\tau_i-d)+S(\tau_{n}, t)\Bigr)+ nw
\end{align*}
for all $t\ge s$. Using the associativity of the convolution,
we write
\begin{align*}
(S\conv \delta_d \conv \delta^{+w})^{(n+1)}\conv S
&=
\Bigl[(S\conv \delta_d \conv \delta^{+w})^{(n)}\conv S
\Bigr] \conv \delta_d\conv \delta^{+w}\conv S\,,
\end{align*}
and then expand
\begin{align*}
(S\conv \delta_d \conv \delta^{+w})^{(n+1)}\conv S(s,t)
&= \min_{s\le \tau\le t}
\left\{
\Bigl[(S\conv \delta_d \conv \delta^{+w})^{(n)}\conv S
\Bigr] (s,\tau-d) + S(\tau,t) \right\} +w\,.
\end{align*}
For $\tau-d\le s+nd$, the
term in the square brackets takes
the constant value $nw$, while the second summand
is non-increasing in $\tau$. Therefore,
the minimum occurs at some $\tau\ge s+(n+1)d$.
By the inductive assumption,
\begin{align*}
\Bigl[(S\conv \delta_d \conv \delta^{+w})^{(n)}\conv S
\Bigr] (s,\tau-d) & = \min_{\tau_o,\dots, \tau_n\in C_n(s,\tau-d)}
\left\{
\sum_{i=1}^n 
\Bigl(S(\tau_{i-1},\tau_i-d)+S(\tau_n , \tau-d)\Bigr)+ nw\right\} \,.
\end{align*}
Since $\tau_o,\dots, \tau_n,\tau\in C_{n+1}(s,t)$,
the claim is proved.
\end{proof}

\section{Variable Bit Rate  Service with Feedback}
\label{sec:cbrvbr}
We next  apply the results from the previous section to a 
specific random service process in a feedback system, 
consisting of a work-conserving  FIFO buffer with a random time-variable service rate, 
which we refer to as variable bit rate (VBR) server. 
The feedback mechanism is as described earlier with window size $w >0$ and feedback delay  $d >0$. 
The VBR server 	offers the service process $S(s,t)=\sum_{k=s}^{t-1} c_k$, 
where $c_k$ is a random amount of available service 
in the $k$-th  time slot.  We assume that the 
$c_k$'s are independent and identically distributed (i.i.d.) random 
variables, with moment-generating functions  $M_c(\theta) = E[e^{\theta c_k}]$. 
The moment-generating function of $S$ is $M_S(\theta,s,t) = \left(M_c(\theta)\right)^{t-s}$, and 
the effective capacity has the 
simple expression $\gamma_S(-\theta)=-\frac{1}{\theta}\log M_c(-\theta)$.

\subsection{Bounds for a VBR service with feedback}
\label{subsec:vbr}

We now derive bounds on the
equivalent service process  $S_{\rm win}$ of a VBR server in the 
feedback system of  Fig.~\ref{fig:feedback-w-d}. 
The bounds will be expressed either in terms of a statistical service curve for 
a chosen violation probability $\eps>0$ (using Eq.~\eqref{eq:Srvc-probbound}) 
or in terms of the effective capacity (from   Eq.~\eqref{eq:Eff-capacity}), 
both of which require bounds on the moment-generating function of $S_{\rm win}$, denoted by 
$M_{S_{\rm win}}$.   
We use the expression  for $\Swin$ given in Eq.~\eqref{eq:swin-w-d} 
as the starting point for the computation of $M_{S_{\rm win}}$, where we  exploit the following  
relationship of  moment-generating functions. 
	\begin{lemma} \label{lemma:mgfofmin}
		Given two  bivariate random processes $f$ and $g$. Then, for every $\theta>0$,  
		\begin{align*}
		M_{f\minim g}(-\theta,s,t)\leq M_{f}(-\theta,s,t)+M_{g}(-\theta,s,t) \, . 
		\end{align*}
Note that this lemma does not require  $f$ and $g$ to be independent. 
	\end{lemma}
	\begin{proof}
		\begin{align*}
			M_{f\minim g}(-\theta,s,t)&=E\big[e^{-\theta\min{\{f(s,t),g(s,t)\}}}\big]\\
			&=E\big[\max{\{e^{-\theta f(s,t)},e^{-\theta g(s,t)}\}}\big]\\
			&\leq E\big[e^{-\theta f(s,t)}+e^{-\theta g(s,t)}\big]\\
			&=M_{f}(-\theta,s,t)+M_{g}(-\theta,s,t) \, . 
		\end{align*}
	\end{proof}

Since the service guaranteed by $S$ is independent 
on disjoint intervals, 
that is $S (t_1, t_2)$ and $S(t_3, t_4)$ are independent for 
$t_1 < t_2 \le t_3 < t_4$, 
standard techniques for bounding
the moment-generating functions of
sums and convolutions
can be applied directly to Eq.~\eqref{eq:swin-w-d}
to obtain, for $\theta>0$,
\begin{align} 
\label{eq:Swin-with-fidlersbound}
\begin{split} 
  M_{S_{\rm win}}(-\theta,s,t) 
& \le \sum_{n=0}^\infty 
\left\{ \sum_{s=\tau_o\le \ldots \le \tau_n\le t}
\left(
\prod_{i=1}^n
M_S(-\theta,\tau_{i-1},\tau_i-d) \right)\cdot M_S(-\theta,
\tau_n,t)e^{-\theta n w}\right\}\\
 & \le \left(M_c(-\theta)\right)^{t-s}
\sum_{n=0}^\infty \dbinom{t-s+1+n}{n}
\left((M_c(-\theta))^{-d} e^{-\theta w}\right)^n\\
& =
\frac{\left(M_c(-\theta)\right)^{t-s}}{\left(
1-(M_c(-\theta))^{-d} e^{-\theta w}\right)^{t-s+2}}\,,
\end{split}
\end{align}
as long as the convergence condition
$$
M_c(-\theta)^{-d}e^{-\theta w}<1 
$$ 
holds. 
The first line follows
by Lemma~\ref{lemma:mgfofmin}, and  the independence of the service
on disjoint intervals in each summand.
In the second line, $M_S$ is expressed in terms of $M_c$;
the binomial coefficient counts
the number of non-decreasing sequences
$\tau_o\le \ldots \le \tau_n$
with $\tau_o=s$ and $\tau_n\le t$.
The last line follows from the identity
$$
\sum_{n=0}^\infty \dbinom{t+n}{n} x^n = \frac{1}{(1-x)^{t+1}}\,,\quad (|x|<1)\,.
$$

Eq.~\eqref{eq:Swin-with-fidlersbound} provides a useful bound only 
when $\theta$ is chosen so that
the convergence condition is satisfied, and the
right-hand side of Eq.~\eqref{eq:Swin-with-fidlersbound}
is less than one.
We now improve the bound with the help of Lemma~\ref{lem:Swin-giganticMin}. 

\bigskip 
\begin{theorem} 
\label{thm:Swin-VBR}
Let $S(s,t)$ be a VBR server with feedback as described above.
Then, for every $\theta>0$,
\begin{align}
\label{eq:Swin-VBR}
M_{S_{\rm win}}(-\theta,s,  t) \le
\left( \bigl(M_c(-\theta)\bigr)^d 
+ d e^{-\theta w}\right)^{\left\lfloor\frac{t-s}{d}
\right\rfloor}\,.
\end{align}
\end{theorem}




\begin{proof}  Fix $\theta>0$. Since $S_{\rm win}(s,t)$ is non-decreasing
in $t$, we can round down the length
of the time interval to the nearest integer multiple of $d$.
Thus, it suffices to consider the case where $t-s=Nd$
for some integer $N$.
By Lemma~\ref{lem:Swin-giganticMin},
\begin{align*}
S_{\rm win}(s,t) =
\bigwedge_{n=0}^N
\left\{ \min_{\tau_o,\dots, \tau_n\in C_n(s,t)}
\sum_{i=1}^{n}
\Bigl(S(\tau_{i-1},\tau_i-d)+S(\tau_n , t)\Bigr)+ nw\right\} \,.
\end{align*}
By Lemma~\ref{lemma:mgfofmin}
and using that $S$ is independent on disjoint time intervals,
we estimate
\begin{align*}
M_{\Swin} (-\theta,s,t) 
&\le \sum_{n=0}^N \left\{
\sum_{\tau_o,\dots, \tau_n\in C_n(s,t)}
\left(\prod_{i=1}^n M_S(-\theta,\tau_{i-1},\tau_i-d)
\right) M_S(-\theta, \tau_n, t)e^{-\theta n w}\right\} \, . 
\end{align*}
The number of sequences in $C_n(s,t)$ is bounded by
$$
|C_n(s,t)|=
\dbinom{(N-n)d + n}{n} = 
\prod_{j=1}^n \frac{(N-n)d+j}{j}
\le \prod_{j=1}^n \frac{(N-n)d+dj}{j}
= d^n \dbinom{N}{n}\,.
$$
Expressing $M_S$ in terms of $M_c$, we arrive at
\begin{align*}
M_{S_{\rm win}}(-\theta,s,t) 
&\le \sum_{n=0}^N 
|C_n(s,t)| \, \bigl(M_c(-\theta)\bigr)^{(N-n)d} e^{-\theta n w}\\
&\le \sum_{n=0}^N \dbinom {N}{n}
\bigl((M_c(-\theta))^d\bigr)^{N-n} 
\bigl(d e^{-\theta w}\bigr)^n\\
&= \bigl((M_c(-\theta))^d + 
d e^{-\theta w}\bigr)^N \,.
\end{align*}
Since $N=\frac{t-s}{d}$, the claim is proved.
\end{proof}

\bigskip
Theorem~\ref{thm:Swin-VBR} (as well as Eq.~\eqref{eq:Swin-with-fidlersbound}) can be directly inserted into expressions of the MGF calculus, e.g., for the backlog expression in 
Eq.~\eqref{eq:backlog-MGF}. Further, via Eq.~\eqref{eq:Srvc-probbound}, the theorem also provides a statistical service curve for $\Swin$, which we will denote as $\S_{\rm win}^\eps$.  
By taking the logarithm of $M_\Swin$, we can obtain bounds on the effective capacity $\gamma_{\rm win}$,    
as expressed in 
this  corollary. 
\begin{corollary}\label{cor:Swin-VBR}
Lower bounds on the effective capacity $\gamma_{\rm win}$ of a VBR process 
$S(s,t)$ with feedback  are given for $\theta >0$ by 
\begin{align}
\label{eq:effC-VBR-unimproved}
  \text{a)} \quad & \gamma_{\rm win}(-\theta) \ge 
\gamma_S(-\theta) +
\frac{1}{\theta} \log \left(1-e^{\theta (d \gamma_S(-\theta)-w)} \right)\,, \\
\label{eq:effC-VBR}
  \text{b)}\quad & 
 \gamma_{\rm win} (-\theta)  \ge 
 \gamma_S(-\theta) 
-\frac{1}{d\theta}\log \left(
1+ d e^{\theta (d \gamma_S(-\theta) -w)}\right) \, . \qquad\qquad\qquad\qquad
\end{align}
\end{corollary}
The first bound follows from Eq.~\eqref{eq:Swin-with-fidlersbound}, and requires the convergence condition
$\gamma_S(-\theta) < \frac{w}{d}$. 
The second bound follows immediately from Theorem~\ref{thm:Swin-VBR}. 
Eqs.~\eqref{eq:effC-VBR-unimproved} and~\eqref{eq:effC-VBR} 
clearly express the service impediment due to the feedback process, by subtracting a positive term from the available service 
without feedback. 
In Subsec.~\ref{subsec:eval-VBR}, we present numerical examples that 
evaluate both bounds. 

All results in this section can be applied to a `leftover
service' model at a server offering
a constant-rate service in the presence of cross-traffic
with independent increments.
The leftover service expresses the service available to
a flow in terms of the capacity that is left unused 
by competing cross-traffic. The leftover model assumes that the 
analyzed traffic flow has lower priority than the cross-traffic. 
Explicitly, if the service rate is $C$ and cross-traffic arrivals are given by
$$
A^{\rm c}(s,t)=\sum_{k=s}^{t-1}a^{\rm c}_k\,,
$$
where the cross-traffic arrivals in each time slot are given by
an i.i.d. sequence of random variables $a^{\rm c}_k$,
then the leftover service $S^{\rm lo}$ available to the flow 
satisfies
$$
S^{\rm lo}(s,t)\ge \sum_{k=s}^{t-1} (C-a^{\rm c}_k)\,.
$$
Although the summands $C-a^{\rm c}_k$ may take negative values,
Theorems~\ref{thm:Swin-apriori} and~\ref{thm:Swin-VBR}
remain valid and provide non-trivial bounds on
the service process with feedback, as long
as the stability condition $E[a^{\rm c}_k]<C$ is satisfied.

\subsection{Quality of the VBR bounds}
\label{subsec:VBR-accuracy}

We next use Theorem~\ref{thm:Swin-apriori}  to address the accuracy of the derived lower bounds from Subsec.~\ref{subsec:vbr}.  
Applying the theorem to  the VBR server gives the  bounds
\begin{align}
\label{eq:Swin-VBR-w/d-x}
\sum_{k=s}^{t-1} \min
\Bigl\{c_k,\frac{w}{d}\Bigr\}\ \le  S_{\rm win}(s,t) 
\le  \min \Bigl\{
\sum_{k=s}^{t-1} c_k,\left\lceil\frac{t-s}{d}\right\rceil
w\Bigr\}\,.
\end{align}
Even though the upper bound is optimistic, its difference to 
the (lower) bounds on the service computed from 
Theorem~\ref{thm:Swin-VBR} 
limits the deviation from the true value of $\Swin$. 
Moreover, the difference between the upper and lower  bounds in Eq.~\eqref{eq:Swin-VBR-w/d-x} 
indicates the range of useful estimates. Note that the 
lower bound in Eq.~\eqref{eq:Swin-VBR-w/d-x} 
is exact for the special case $d=1$. In particular, 
any  lower bound that falls below 
this bound can be replaced by the simpler estimate 
of  Eq.~\eqref{eq:Swin-VBR-w/d-x}. 

Let us, for the moment, consider a deterministic server with $c_k \equiv C$. Then, 
the lower and upper bounds are essentially equivalent, giving 
$ S_{\rm win}(s,t) \approx \min \{C ,\frac{w}{d}\} (t-s)$. This corresponds to a 
well-known expression for the throughput of a window flow control system, e.g., \cite[Eq. 6.1]{Book-BeGa}. On the other hand, when $S$ is random, the upper and lower bounds have different
long-term average rates.

We will use Eq.~\eqref{eq:Swin-VBR-w/d-x} to obtain upper and lower bounds on $\S_{\rm win}^\eps$ and $\gamma_{\rm win}$. 
We obtain a lower bound on $\S_{\rm win}^\eps$ by taking the moment-generating 
function of the lower bound in Eq.~\eqref{eq:Swin-VBR-w/d-x}, resulting for 
$\theta >0$ in  
\begin{align}
\label{eq:Swin-VBR-w/d-x1}
M_{S_{\rm win}}(-\theta,s,t) 
\ \le \ \bigl(E[e^{-\theta\min\{c_k,\frac{w}{d} \}}]\bigr)^{t-s}\,. 
\end{align}
Note the change of  direction  of the inequality since we  
use $-\theta$ as a function argument.  Inserting this bound into 
Eq.~\eqref{eq:Srvc-probbound} provides a lower  bound on $\S_{\rm win}^\eps(s,t)$.

Obtaining an opposing upper bound on $\S_{\rm win}^\eps$ from 
 Eq.~\eqref{eq:Swin-VBR-w/d-x} is not as straightforward, since 
statistical 
service curves express lower bounds on the available service. 
We exploit that, for some VBR servers, it is possible to get the exact distribution of 
$S(s,t) = \sum_{k=s}^{t-1} c_k$, where $S$ corresponds  to the 
service process of the VBR system without feedback. For example, 
in  Subsec.~\ref{subsec:eval-VBR}, where we use an exponentially  
distributed $c_k$,  the service $ S(s,t)$ has an Erlang distribution. In such cases, since 
the upper bound in Eq.~\eqref{eq:Swin-VBR-w/d-x} implies $\Swin \le S$, the $\eps\,$--quantile   
of $S$ is an upper bound for any statistical service curve  $\S_{\rm win}^\eps$ 
of the VBR server with feedback. 
Thus, 
Eq.~\eqref{eq:Swin-VBR-w/d-x} gives 
\begin{align}
\label{eq:Swin-upperbound}
\S_{\rm win}^\eps(s,t) \le 
\min \bigl \{\eps\,\text{--quantile   
of } S (s,t),
\left\lceil\frac{t-s}{d}\right\rceil w\bigr\}\, ,  
\end{align}
where we also use the second (deterministic) term of the upper bound 
in Eq.~\eqref{eq:Swin-VBR-w/d-x}. 

Bounds for the  effective capacity  are 
directly obtained from  Eq.~\eqref{eq:Swin-VBR-w/d-x} by first computing moment-generating functions and then taking the logarithm.
\begin{corollary}\label{cor:apriori-VBR}
Under the assumptions of Theorem~\ref{thm:Swin-VBR}, $\gamma_{\rm win} (-\theta)$ 
is bounded for $\theta>0$ by 
\begin{align*}
-\frac{1}{\theta} 
\log E\bigl [e^{-\theta \min\left\{c_k,\frac w d\right\}}\bigr]
 \le \gamma_{\rm win}(-\theta)  \le 
\min \left\{
\gamma_S(-\theta), \frac{w}{d}\right\}\,.
\end{align*}
\end{corollary}
The lower bound 
becomes an equality when $d=1$.
The upper bound is strictly larger than the lower bound for all values of $\theta> 0$.
Most importantly, the upper bound becomes sharp if we take $d\to\infty$ while holding $w/d$ fixed.
In the deterministic case, with $c_k\equiv C$, 
the upper and lower bounds both equal 
$\min\{C,\frac w d \}$, in accordance with our previous discussion. 
Moreover, in the limit $\theta\to\infty$, 
the right-hand side of Eq.~\eqref{eq:effC-VBR} 
is asymptotic to the shown bounds.
Therefore, the bounds on the effective capacity are sharp in
the limit $\theta\to \infty$.

\subsection{Numerical evaluation of VBR bounds}
\label{subsec:eval-VBR}

We now present a numerical evaluation of our bounds for the VBR server with feedback. 
We consider a VBR server with an exponential distribution 
with 
moment-generating function and effective capacity given by 
\begin{align*}
M_S(-\theta,s,t)=(1+C\theta)^{-(t-s)}  \, , \
\gamma_S(-\theta)= \frac{1}{\theta} \log(1+C\theta)\,.
\end{align*} 
We select $C=1$~Mb for the available service in a time slot of length 1~ms, 
which gives an average service rate of 1~Gbps.


\begin{figure}[!t]
\centering
	\subfigure[$w/d= 100$~Mbps.]{
 \includegraphics[width=0.8\textwidth]{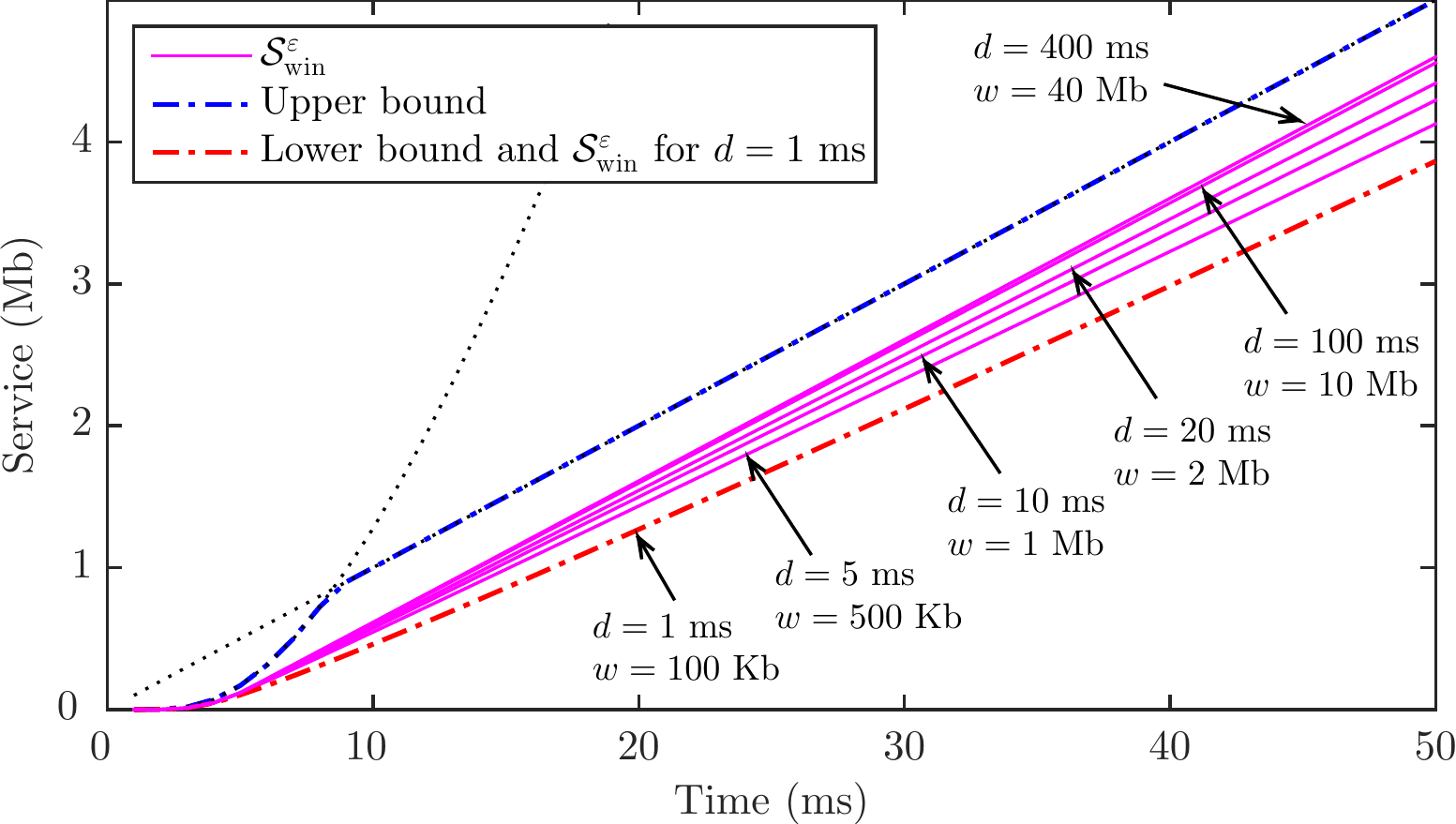}
		\label{fig:VBR-Swin-100}
	}
	
\centering
	\subfigure[$w/d= 500$~Mbps.]{
		\includegraphics[width=0.8\textwidth]{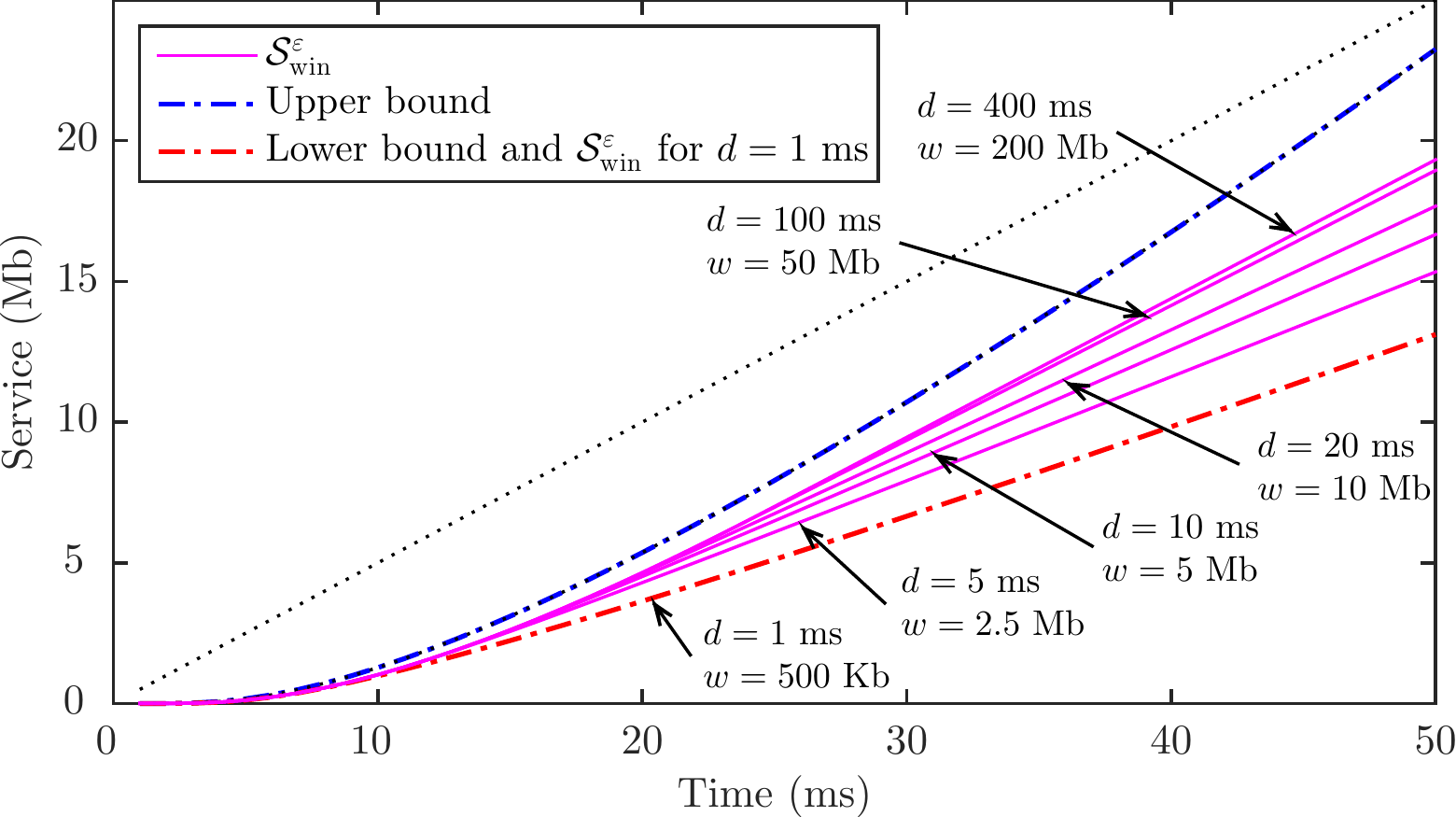}
		\label{fig:VBR-Swin-500}
	}
\caption{Statistical service curves $\S_{\rm win}^\eps (0,t)$ for VBR server with feedback (Avg. rate: 1~Gbps, $\eps = 10^{-6}$). \small{The dotted lines correspond to the two terms of the upper bound 
in Eq.~\eqref{eq:Swin-upperbound}.}}
\label{fig:VBR-Swin}

\end{figure}
\begin{figure}[!t]
\centering
	\subfigure[$w/d= 100$~Mbps.]{
 \includegraphics[width=0.8\textwidth]{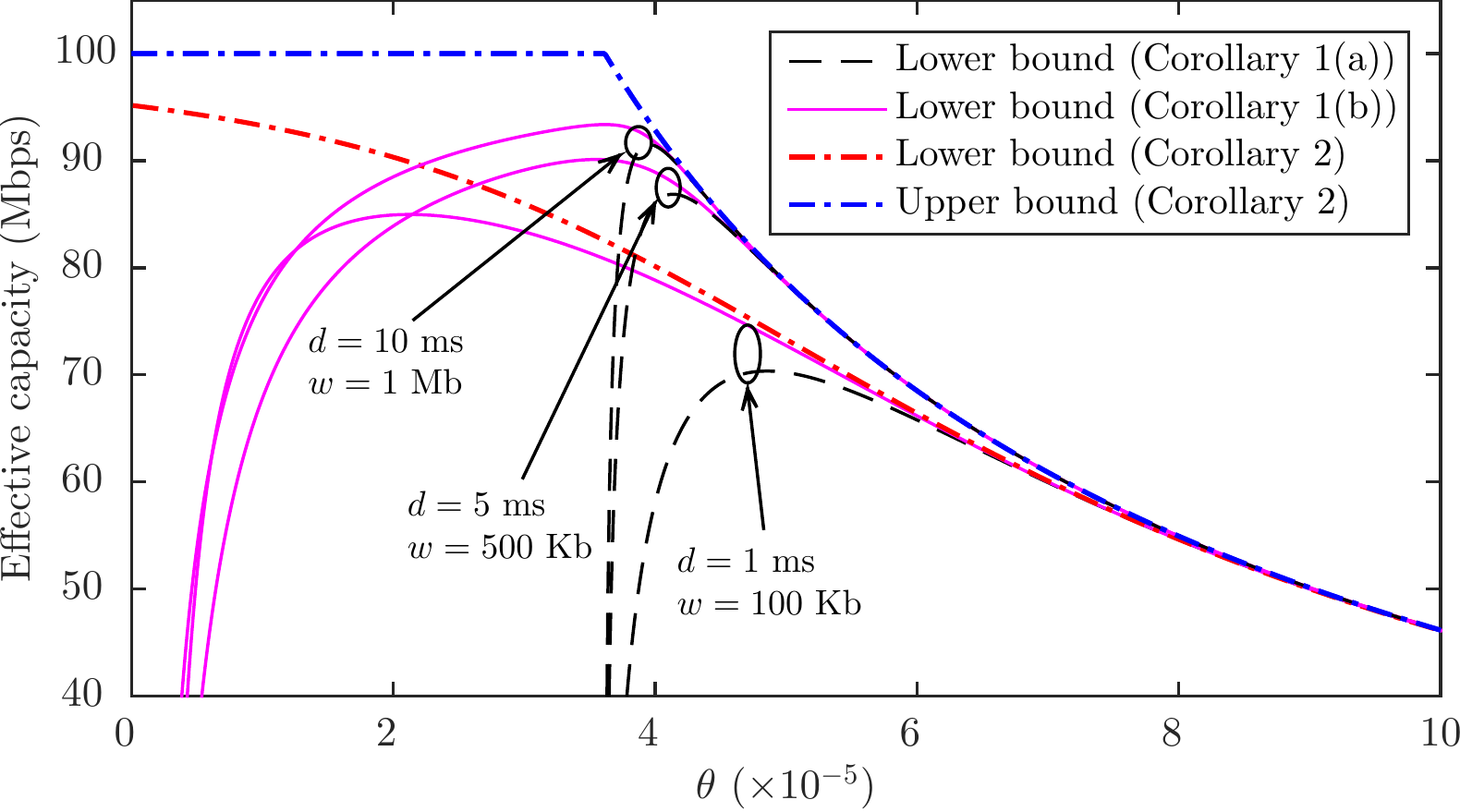}
		\label{fig:VBR-EffCap-100}
	}
		
\centering
	\subfigure[$w/d= 500$~Mbps.]{
\includegraphics[width=0.8\textwidth]{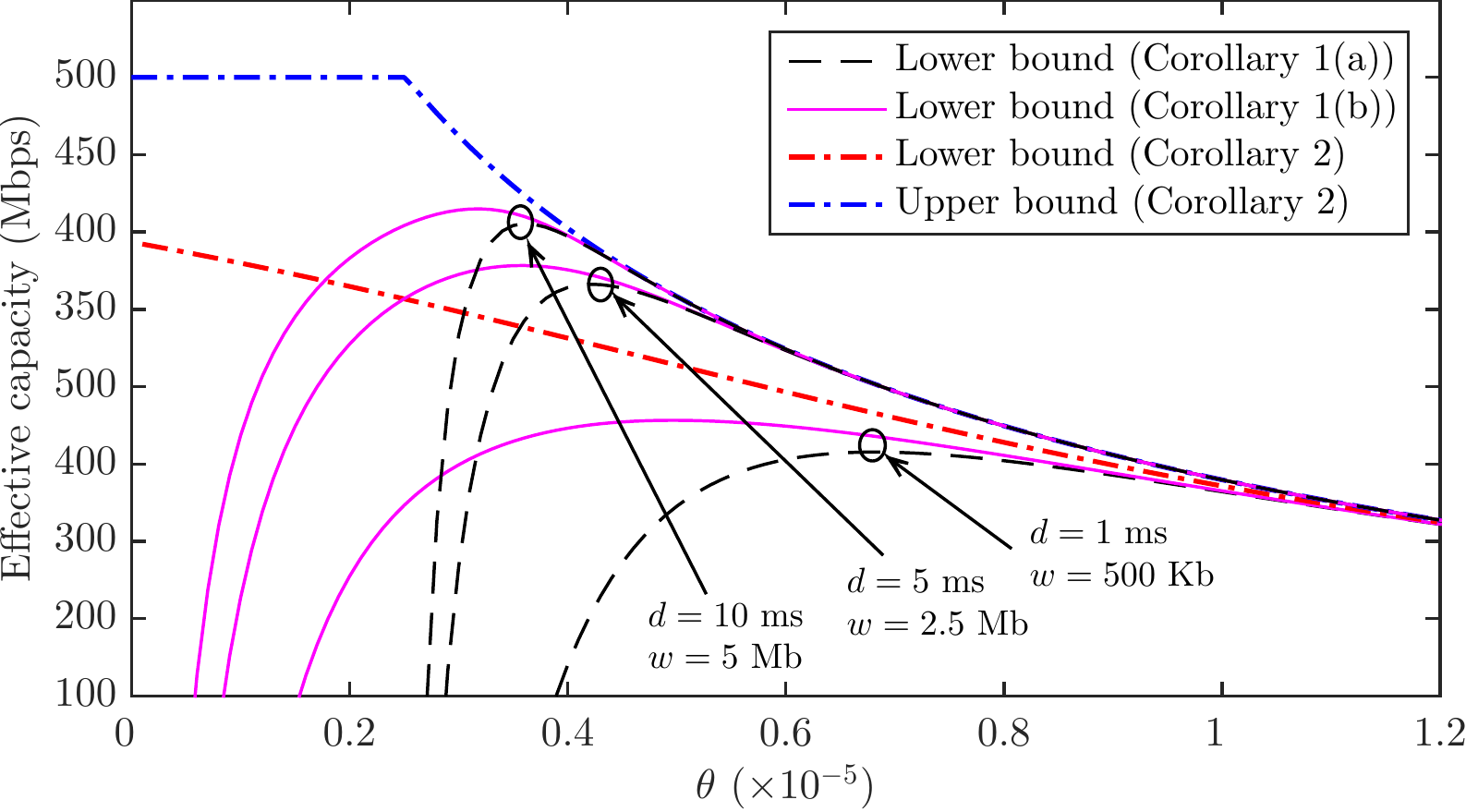}
		\label{fig:VBR-EffCap-500}
	}			
\caption{Effective capacity $\gamma_{\rm win}(\theta)$ of VBR server with feedback (Avg. rate: 1~Gbps).}
\label{fig:VBR-EffCap}

\end{figure}

In Figs.~\ref{fig:VBR-Swin-100} and~\ref{fig:VBR-Swin-500} we plot  statistical service curves  $\S_{\rm{win}}^\eps (0,t)$ as functions  of time, 
where we use a violation probability of $\eps=10^{-6}$.
We compute multiple service curves where we 
vary the delay $d$ and the window size $w$, but fix the ratio $w/d$.
The statistical service curves use Eq.~\eqref{eq:Srvc-probbound} with the 
bound on $M_\Swin$ from Theorem~\ref{thm:Swin-VBR}. We compare these service curves 
to the upper and lower bounds computed from 
Eq.~\eqref{eq:Swin-VBR-w/d-x}.

Fig.~\ref{fig:VBR-Swin-100} shows  statistical service curves, plotted as solid 
lines, using a ratio $w/d = 100$~Mbps.  
The upper and lower bounds are represented by dash-dotted lines. We also  
indicate the  
two terms in the upper bound of Eq.~\eqref{eq:Swin-upperbound} by dotted lines. 
Note that the  lower bound is in fact equal to 
$\S_{\rm win}^\eps$ for the special case $d=1$~ms, due to 
Lemma~\ref{lem:Swin-d1-exact}. 

It is evident that the rates of the service curves match well to those of the upper and lower bounds. 
We observe that increasing the 
delay and the window size simultaneously 
improves the available 
service. 
For large values of $d$, the service curves appear to have a supremum 
well below the plotted upper bound. 
This is expected, since 
the rates of our bounds are exact for $d \to \infty$. 
In a deterministic feedback system with a fixed rate of 1~Gbps, 
varying $w$ and $d$ with a fixed ratio 
$w/d$ results in all cases in an  essentially constant rate service of 
100~Mbps, with minuscule deviations.  

Fig.~\ref{fig:VBR-Swin-500} evaluates the same scenario for a different parameter selection, this time, fixing   $w/d = 500$~Mbps. Note that the range of 
the y-axis is modified from  Fig.~\ref{fig:VBR-Swin-100}. For the shown range of 
time values, the rate of the 
statistical service curves are close to that of the lower bound, but noticeably smaller 
than that of the upper bound. This will change when we consider larger time intervals.  
The reason is that with  this choice of $w/d$,
the upper bound 
of the service curves is dominated for a longer period of time  by the first term  of Eq.~\eqref{eq:Swin-upperbound}. Once the second term (with rate $w/d$) governs 
the bound, the rate of the statistical service curves will be close to $w/d$ as well.

 We now turn to the effective capacity. 
For the same set of parameters  as before, we evaluate in Fig.~\ref{fig:VBR-EffCap} the 
lower bounds of the effective capacity $\gamma_{\rm win}$ from Corollary~\ref{cor:Swin-VBR}. 
Note that the corollary  presents two bounds. Corollary~\ref{cor:Swin-VBR}(a) 
(Eq.~\eqref{eq:effC-VBR-unimproved}) does not take advantage of Theorem~\ref{thm:Swin-VBR}, whereas Corollary~\ref{cor:Swin-VBR}(b) (Eq.~\eqref{eq:effC-VBR}) involves Theorem~\ref{thm:Swin-VBR}.
We plot these bounds on $\gamma_{\rm win}(-\theta)$ as a function of $\theta >0$. Recall 
that the actual (not estimated) effective capacity $\gamma_{\rm win}(-\theta)$ is a decreasing function of  $\theta$, and that its value for $\theta \to 0$ is the 
average service rate. 

Fig.~\ref{fig:VBR-EffCap-100} depicts the  bounds for 
$w/d = 100$~Mbps, and Fig.~\ref{fig:VBR-EffCap-500} those for  
$w/d = 500$~Mbps. Consider first the lower and upper bounds 
from Corollary~\ref{cor:apriori-VBR},  which are indicated by dash-dotted lines. 
As discussed in Sec.~\ref{subsec:VBR-accuracy},  the bounds converge for large values of $\theta$. The figure 
indicates that the convergence occurs early. 
We observe that bounds from  Corollary~\ref{cor:Swin-VBR}(a) are inferior to 
those of Corollary~\ref{cor:Swin-VBR}(b), which emphasizes the value of applying Theorem~\ref{thm:Swin-VBR}. 
For $d>1$~ms, the lower bounds  of Corollary~\ref{cor:Swin-VBR}(b) for different values of $w$ and $d$ are accurate when $\theta> 4 \cdot 10^{-5}$ (Fig.~\ref{fig:VBR-EffCap-100}) and $\theta> 0.5 \cdot 10^{-5}$ (Fig.~\ref{fig:VBR-EffCap-500}), but degrade for small values of  $\theta$. When this happens the lower bound from Corollary~\ref{cor:apriori-VBR} should be used. 
For $d=1$~ms, the bounds from Corollary~\ref{cor:Swin-VBR}(b) are  pessimistic over 
a large range. Here, the lower bound from Corollary~\ref{cor:apriori-VBR}, 
which is exact when $d=1$~ms, is the better result. 
Since the actual $\gamma_{\rm win}(-\theta)$ increases 
when reducing $\theta$, the best estimates of the average available service of $\Swin$ 
are obtained at the maximum of the curves. 
It is interesting to observe that Corollary~\ref{cor:Swin-VBR} generates bounds 
which are not always monotonic when increasing $d$.


\begin{figure}[!t]
\centering
		\includegraphics[width=0.8\textwidth]{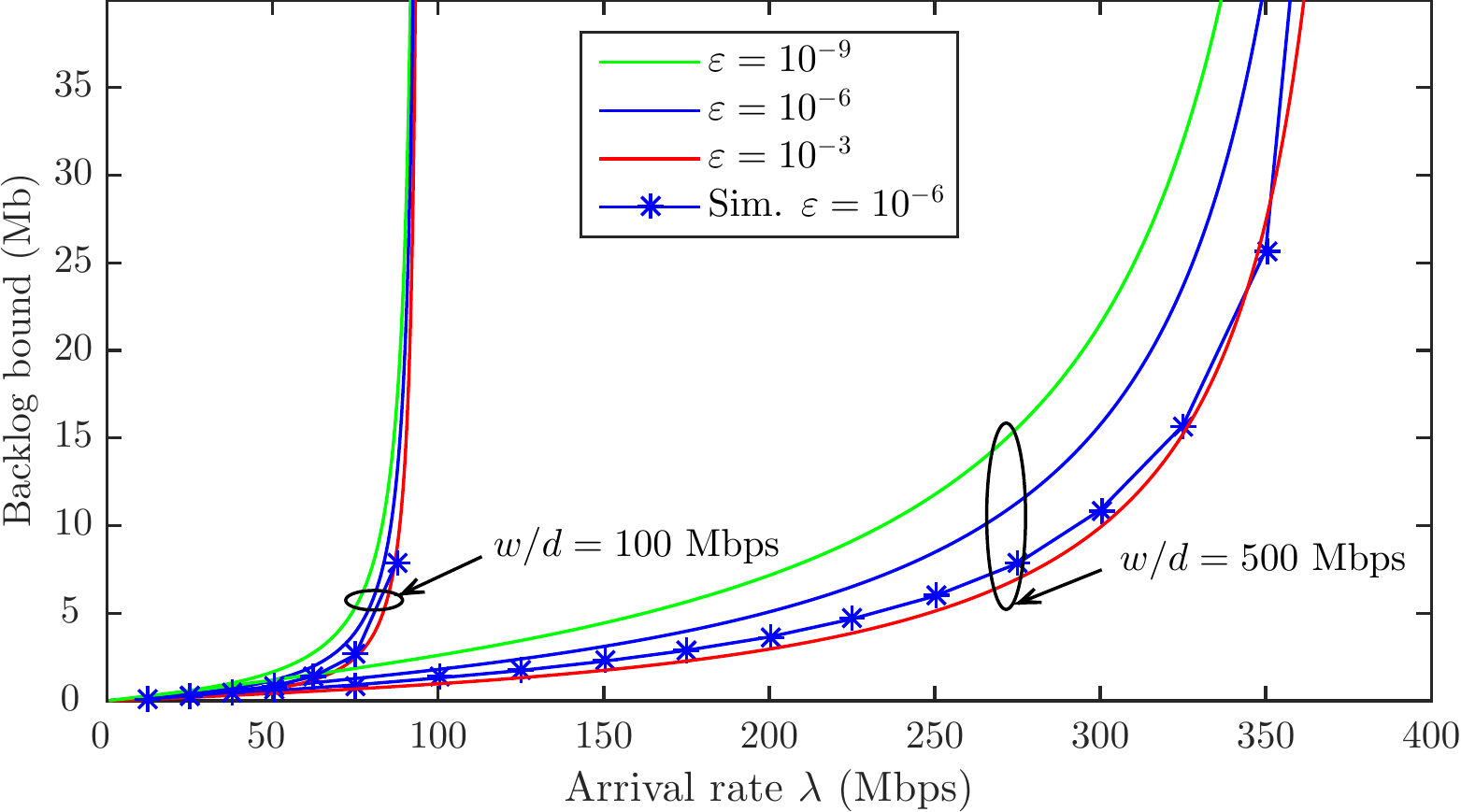}
		\caption{Backlog bounds for VBR server with feedback ($C= 1$~Gbps).}
\label{fig:VBR-Backlog}

\end{figure}


Our  last numerical example presents a backlog analysis of the VBR server 
with feedback, by applying Eq.~\eqref{eq:backlog-MGF}. 
For the arrivals, we select a process similar to the service process, 
where arrivals in each time slot are i.i.d. with an exponential distribution
that generates arrivals at an average rate of $\lambda$~Mbps. The moment-generating function of this arrival process, 
denoted by $M_A$, is given by 
\begin{align*}
M_A(\theta,s,t)=\frac{1}{(1-\lambda\theta)^{t-s}} \,.
\end{align*}
The service is the same VBR server as evaluated before. 
By setting $d=1$~ms, we can 
take advantage of 
the exact expression for $M_\Swin$ that follows from Lemma~\ref{lem:Swin-d1-exact}. Even though we have exact 
moment-generating functions for arrivals and service, the backlog bound 
of the MGF network calculus will be pessimistic, due to the deconvolution expression in Eq.~\eqref{eq:mgf-deconv}. 
In Fig.~\ref{fig:VBR-Backlog} we depict backlog bounds for the ratios 
$w/d = 100$~Mbps and $500$~Mbps, for different violation probabilities $\eps$, 
as functions  of the arrival rate $\lambda$. 
We note that  this presents the first probabilistic backlog bounds of 
a random  feedback system using methods of the network calculus.
We observe that 
for each choice of $w$, the system saturates at a well-defined rate. 
The saturation rates are 
close to the lower bounds on $\gamma_{\rm win} (-\theta)$ for 
$\theta \to 0$ in Fig.~\ref{fig:VBR-EffCap}. 
The plots indicate a low sensitivity of the bounds to $\eps$. 
We also  include simulation results for $\eps=10^{-6}$.\footnote{The simulations 
represents runs of $10^{9}$ time slots, where the simulation is started with an empty system, and  results from first $10^5$ time slots are discarded.} A comparison with the  simulations 
shows that the achieved backlog bounds are pessimistic, but 
track the blow up of the backlog at high utilizations well. 
The analytical bounds will be more pessimistic for  $\eps=10^{-3}$  
and less so for  $\eps=10^{-9}$.

\section{Service Processes with Positive Time Correlations}
\label{sec:MMOO}

The VBR server from the previous section offered an i.i.d. service in 
each time slot, and the time correlations of the feedback system resulted 
exclusively from the feedback mechanisms. 
Nonetheless, the analysis of this simple feedback system 
proved to be substantial. 
This raises the question whether feedback systems
with more complex service processes are at all tractable with  
our analysis approach.
In this section, we provide a positive answer, by analyzing  
a feedback system for a service process with memory. 

The server is represented by a Markov-modulated 
On-Off (MMOO) process \cite{Book-Chang}, which 
operates in two states. In the ON state (state~1) 
the server is transmitting a constant 
amount of  $P>0$ units of traffic per time slot. 
In the OFF state (state~0), the server does not 
transmit any traffic.
The state is selected at the beginning of each time slot using 
fixed transition probabilities, 
where $p_{ij}$ denotes the probability of moving
from state $i$ to state~$j$ ($i,j \in \{0,1\}$).
By definition, the service process is additive, 
$$
S(s,t) =\sum_{k=s}^{t-1} c_k\,,
$$
where $c_k=P$ if the system is in the ON state at
time $k$, and $c_k=0$ otherwise.
Since $p_{01}= 1  -p_{00}$ and 
$p_{10}= 1 - p_{11}$, the system
is fully characterized by
the three parameters, $p_{00}$, $p_{11}$,
and $P$.  The Markov chain is assumed
to be in its steady-state, where
the probability that the system is ON in any given
time slot equals $\P(ON)=\frac{p_{01}}{p_{10}+p_{01}}$.
MMOO processes are frequently used in the literature
for modelling bursty traffic or service \cite{KeWaCh93}.
The average rate of the MMOO process is 
$$
\frac{ E[S(s,t)]}{t-s} = \frac{p_{01}}{p_{01}+p_{10}}P\,.
$$
Its effective capacity 
is given in \cite[Eq.~(7.18)]{Book-Chang} as
\begin{equation}
\label{eq:MMOO-effC}
\gamma_S (-\theta) = -\frac{1}{\theta}
\log  \left( \frac12 \left\{ (p_{00}+p_{11}e^{-\theta P})
+ \sqrt{(p_{00}+p_{11}e^{-\theta P})^2-4
(p_{00}+p_{11}-1)e^{-\theta P}} \right\} \right)
\end{equation}
for all $\theta\ge 0$.

We assume throughout this section that
the transition probabilities satisfy
$p_{01}+p_{10} <1$. 
This condition ensures that the system does not alternative
rapidly between the two states. Smaller
values of $p_{01}$ cause the system to linger in the
OFF state, while smaller values of $p_{10}$ cause
extended bursts.  

\subsection{Time-correlation properties of MMOO processes}

The analysis of an MMOO server with feedback requires a 
deeper inspection of the properties of an MMOO process. 
The following lemma concerns the time correlations of
the underlying two-state Markov chain. In this subsection 
only, we allow the parameter $\theta$ to take
both positive and negative values.

\bigskip 
\begin{lemma} \label{lem:On-Off-monomials} 
Consider a two-state Markov chain with $p_{01}+p_{10} <1$. 
For every strictly increasing sequence
$\tau_1,\dots ,\tau_n$, the probability that
the system is in the ON state at times $\tau_1,\dots,\tau_n$ is 
a decreasing function of the differences $\tau_i-\tau_{i-1}$,
$i=2,\dots, n$. 
\end{lemma}

\bigskip

\begin{proof} 
Consider the state variables 
$$
X_k = \left\{\begin{array}{ll}
1\quad &\text{if the state is ON at time $k$},\\
0 \quad &\text{else}\,.
\end{array}\right.
$$
The product $\prod_j X_{\tau_j}$
is the indicator of the event that  the system is in  the
ON state at each time $\tau_i$, for $i=1,\dots, n$.
By the Markov property,
\begin{align*}
\P\left(\text {ON at time $\tau_i$ for all $i=1,\dots, n$}\right)
&= E\left [\prod_{i=1}^n X_{\tau_i}\right] \\
&= 
\P(X_{\tau_1}=1) \cdot \prod_{i=2}^n 
E\left [X_{\tau_i}\big\vert X_{\tau_{i-1}}=1\right]\,.
\end{align*}
By stationarity, the leading factor is a constant
determined by the steady state of the Markov chain,
and the $i$-th factor in the product depends
on $\tau_i-\tau_{i-1}$.  
We now verify that each of these factors decreases
with $\tau_i-\tau_{i-1}$.

Let $p=\P(X_k=1)$
be the probability that the state is ON at time $k$,
and let $\tau=\tau_i-\tau_{i-1}$.
We compute the $\tau$-step transition matrix as
\begin{align}
\left(\begin{array}{cc}
p_{00}& p_{01}\\
p_{10} & p_{11}\end{array}\right)^\tau = 
\left(\begin{array}{cc}
(1\!-\!p) & p \\
(1\!-\!p) & p \end{array}\right)
+ \mu^\tau \left(\begin{array}{cc}
p & -p\\
-(1\!-\!p) & (1-p)\end{array}\right)\,,
\end{align}
where the first matrix on the right-hand side is the spectral projection
onto the steady state, 
$\mu=1-p_{01}-p_{10}$ is the non-trivial 
eigenvalue of the transition matrix,
and the second matrix
is the spectral projection onto the eigenstate
corresponding to $\mu$. Since $0<\mu<1$, 
$$
E\left[X_{\tau_i}\ \big\vert \ X_{\tau_{i-1}}=1\right] =
p  + (1-p) \mu^\tau
$$
is a decreasing function of $\tau$, proving the claim.
\end{proof}

The next lemma provides bounds on the moment-generating 
function of the MMOO service process.

\bigskip
\begin{lemma} 
\label{lem:pos-corr}
Let $S$ be an MMOO process as described above. If $p_{01}+p_{10}<  1$,
then the following inequalities hold for all $\theta\in\mathbb{R}$:
\begin{enumerate} 
\item For every strictly increasing sequence
$\tau_1<\dots < \tau_n$, 
\begin{align}
\label{eq:pos-corr-gaps}
E\left[e^{\theta \sum_{i=1}^n c_{\tau_i}}\right]
\le M_S(\theta,0,n)\,.
\end{align}

\item $M_S(\theta, 0, t)$ is supermultiplicative in $t$,
\begin{align}
\label{eq:pos-corr-supermult}
M_S(\theta,0,s) \cdot M_S(\theta,0,t) \le
M_S(\theta,0,s+t)  
\quad (\forall s,t\ge 0)\,.
\end{align}

\item The moment-generating function is bounded by
\begin{align}
\label{eq:pos-corr-effC}
(M_c(\theta))^t\le M_S(\theta,0,t) \le m_+(\theta)^t\,
\quad (\forall t\ge 0)\,,
\end{align}
where $M_c(\theta)=(1\!-\!p)+pe^{\theta P}$ is the
moment-generating function of the service in a single
time slot, and 
where $m_+(\theta)$ is the larger eigenvalue of the matrix
\begin{equation}
\label{eq:MMOO-effC-L}
L(\theta) = \left(\begin{array}{cc} p_{00} & p_{01}\\ p_{10} & p_{11}
\end{array}\right) 
\left(\begin{array}{cc} 1 & 0\\ 0 & e^{\theta P}
\end{array}\right) \,.
\end{equation}

\item Furthermore,
\begin{align}
\label{eq:pos-corr-effC-lower}
M_S(\theta,0,t) \ge 
K(\theta)(m_+(\theta))^t\,,
\end{align}
where $0<K(\theta)<1$ is an explicit constant 
(to be computed in the proof).
\end{enumerate}
\end{lemma}

\bigskip\begin{proof} \begin{enumerate} \item 
Consider first the case $\theta\ge 0$, and
let $X_k$ be the indicator function that the
system is in the ON state at time $k$.  Writing
$$
c_k=PX_k\,,\quad e^{\theta c_k} = 1 + (e^{\theta P}-1)X_k\,,
$$
we expand
$$
e^{\theta \sum_{i=1}^n c_{\tau_i}}
= \prod_{i=1}^n \left(1+ (e^{\theta P}-1)X_{\tau_i}\right)
= \sum_{J\subset \{\tau_1,\dots, \tau_n\}} (e^{\theta P}-1)^{|J|}
\prod_{j\in J} X_j\,.
$$
For each subset $J\subset \{\tau_1,\dots, \tau_n\}$,
the distance between consecutive elements
increases with the distances $\tau_i-\tau_{i-1}$ for
$i=2,\dots, n$, and all coefficients are
positive. Therefore, we can apply Lemma~\ref{lem:On-Off-monomials}
to see that
$$
E\left[e^{\theta \sum_{i=1}^n c_{\tau_i}}\right]
= \sum_{J\subset \{\tau_1,\dots, \tau_n\}} (e^{\theta P}-1)^{|J|}
E\left[\prod_{j\in J} X_j\right]
$$
is a decreasing function of $\tau_i-\tau_{i-1}$. 
Since these differences take the smallest possible value
when $\{\tau_1,\dots,\tau_n\}=\{0,\dots, n-1\}$,
this proves the claim for $\theta\ge 0$.

For $\theta<0$, let $Y_k=1-X_k$ be the indicator 
function that the state is OFF at time $k$,
and set $\phi=-\theta>0$.  We write
$$
e^{\theta c_k} = e^{-\phi P(1-Y_k)} = 
e^{-\phi P}\left(1+ (e^{\phi P}-1)Y_k\right)\,,
$$
expand the product as a sum
$$
e^{\theta \sum_{i=1}^n c_{\tau_i}}
= e^{-n \phi P}\sum_{J\subset \{\tau_1,\dots, \tau_n\}} 
(e^{\phi P}-1)^{|J|}
\prod_{j\in J} Y_j\,,
$$
and argue as in the other case.

\item 
Fix $\theta\in\mathbb{R}$.
For integers $\ell\ge 0$,
let $f(\ell)=E\left[e^{\theta(S(0,s) + S(s+\ell, s+t+\ell))}\right]$.
Factoring the exponential as in the proof of 
Eq.~\eqref{eq:pos-corr-gaps}, it follows from
Lemma~\ref{lem:On-Off-monomials}
that $f(\ell)$
decreases with $\ell$. Taking
$\ell\to\infty$ and $\ell=0$
yields Eq.~\eqref{eq:pos-corr-supermult}.

\item By Part 2, the function
$g(t)=\log M_S(\theta,0,t)$ is  superadditive.
Therefore, the ratio $\frac1t g(t)$ decreases monotonically
from $g(1)=\log M_c(\theta)$ to
$\lim_{t\to\infty} \frac 1t g(t)= \log m(\theta)$, 
and Eq.~\eqref{eq:pos-corr-effC} follows.

\item We start from
\cite[Eq.~(7.15)]{Book-Chang}, which states that
$$
M_S(\theta,0,t) = 
\bigl(
(1\!-\!p)~, ~p \bigr) \, \bigl(L(\theta)\bigr)^t\, 
\binom{1}{1}\,.
$$
Let $m_+(\theta)$ and $m_-(\theta)$ be the larger and
smaller eigenvalues of $L(\theta)$, respectively.
Inserting the spectral decomposition of
$L(\theta)$ with respect to its eigenvalues $m_+(\theta)$ and $m_-(\theta)$ 
into the expression for $M_S(\theta,0,t)$,
we obtain a constant $K(\theta)$ such that 
\begin{equation}
\label{eq:MMOO-effC-lower-proof}
M_S(\theta,0,t) = K(\theta) (m_+(\theta))^t + (1-K(\theta))
(m_-(\theta))^t
\quad (\forall t\ge 0,\theta\in\mathbb{R})\,.
\end{equation}
Using that $M_S(\theta,0,1)=M_c(\theta)$, we see that 
\[ 
K(\theta) = \frac{M_c(\theta)-m_-(\theta)}
            {m_+(\theta)-m_-(\theta)}\,.
\] 
Since $m_+(\theta)>M_c(\theta)>m_-(\theta)$,
it follows that $0<K(\theta)<1$, and we can drop
the second summand in Eq.~\eqref{eq:MMOO-effC-lower-proof}
to obtain Eq.~\eqref{eq:pos-corr-effC-lower}. 
\end{enumerate}
\end{proof}

\bigskip {\em Remark.} 
The representation for $M_S(\theta)$
in Eq.~\eqref{eq:MMOO-effC-lower-proof} implies in
particular that
$$
\gamma_S(\theta)= \lim_{t\to\infty} 
\frac{1}{\theta t} \log M_S(\theta,0,t) = \frac{1}{\theta}\log m_+(\theta)\,,
$$
in agreement with Eq.~\eqref{eq:MMOO-effC}.

The validity of the bounds on the moment-generating function 
and the effective capacity 
extends to a broader class of time-homogeneous 
two-state Markov-modulated processes, 
as described in
\cite[Example 7.2.7]{Book-Chang}. 
The Markov-modulated process is given by a two-state Markov chain
in the steady-state 
and two independent sequences of i.i.d. random variables, 
$c_k^0$ and $c_k^1$.
The service rate in the $k$-th time slot is given by
\begin{equation}
\label{eq:def-MMP}
c_k = (1-X_k) c_k^0 + X_k c_k^1\,,
\end{equation}
where $X_k$ is the state variable of the system.
The service process $S(s,t)$ is defined by $S(s,t)=\sum_{k=s}^{t-1} c_k$.

\begin{lemma}
\label{lem:pos-corr-MMP} Let $S$ be the
two-state Markov-modulated process described above. Assume that
the random variables $c_k^0$ and $c_k^1$
have moment-generating functions 
$M_{c^0} (\theta)= e^{\theta c_k^0}$  and
$M_{c^1} (\theta)= e^{\theta c_k^1}$.
If $p_{01}+p_{10}<1$, then the conclusions of 
Lemma~\ref{lem:pos-corr} hold, with 
$M_c(\theta)=(1-p)M_{c^0}(\theta) + pM_{c^1}(\theta)$,
and
\begin{equation}
\label{eq:MMP-effC-L}
L(\theta) = \left(\begin{array}{cc} p_{00} & p_{01}\\ p_{10} & p_{11}
\end{array}\right) 
\left(\begin{array}{cc} M_{c^0}(\theta) & 0\\ 0 &
M_{c^1}(\theta)
\end{array}\right) \,.
\end{equation}
\end{lemma}

\begin{proof} 
Fix $\theta\in\mathbb{R}$,
and assume without loss of generality
that $M_{c^0}(\theta)\le M_{c^1}(\theta)$.
We follow the proof of Lemma~\ref{lem:pos-corr}.

For the first claim, we condition the left hand
side of Eq.~\eqref{eq:pos-corr-gaps}
on the state of the system at the
times $\tau_1,\dots, \tau_n$.
In a single time slot, we obtain from Eq.~\eqref{eq:def-MMP} that
$$
E\left[e^{\theta c_k}\big\vert X_k\right]
= M_{c^0}(\theta) +\bigl(M_{c^1}(\theta)-M_{c^0}(\theta)\bigr)X_k\,.
$$
Since the random variables $c_k^0$, $c_k^1$ and $X_k$
are all independent, it follows that
\begin{align*}
E\left[ e^{\theta \sum_{i=1}^n c_{\tau_i}}\Big\vert
X_{\tau_1},\dots, X_{\tau_n} \right]
&= \prod_{i=1}^n M_{c^0}(\theta) +
\bigl(M_{c^1}(\theta)-M_{c^0}(\theta)\bigr)X_{\tau_i}\\
& = \sum_{J\subset\{\tau_1,\dots, \tau_n\}}
\bigl(M_{c^0}(\theta)\bigr)^{n-|J|}
\bigl(M_{c^1}(\theta)-M_{c^0}(\theta)\bigr)^{|J|}
\prod_{j\in J} X_j\,.
\end{align*}
Taking expectations, we obtain
$$
E\left[ e^{\theta \sum_{i=1}^n c_{\tau_i}}\right]
= \sum_{J\subset\{\tau_1,\dots, \tau_n\}}
\bigl(M_{c^0}(\theta)\bigr)^{n-|J|}
\bigl(M_{c^1}(\theta)-M_{c^0}(\theta)\bigr)^{|J|}
E\left[\prod_{j\in J} X_j\right]\,.
$$
Since all coefficients are non-negative, 
Eq.~\eqref{eq:pos-corr-gaps} now follows from 
Lemma~\ref{lem:On-Off-monomials}.

A similar application of Lemma~\ref{lem:On-Off-monomials}
yields Eq.~\eqref{eq:pos-corr-supermult},
which directly implies
Eq.~\eqref{eq:pos-corr-effC}.
Finally, Eq.~\eqref{eq:pos-corr-effC-lower} follows
as before from \cite[Eq.~(7.15)]{Book-Chang}.
\end{proof}

\subsection{Bounds for an MMOO server with feedback}

The results in the previous subsection imply 
bounds on the equivalent service  of 
a feedback system in Fig.~\ref{fig:feedback-w-d} containing an 
MMOO server. 
We assume $d>0$ and $w>0$ as  parameters of the system. 
Our first result is analogous to Theorem~\ref{thm:Swin-VBR}.

\bigskip

\begin{theorem}
\label{thm:Swin-MMOO}
Let $S(s,t)$ be an MMOO service process with feedback.  
If the transition probabilities satisfy $p_{01}+p_{10}<  1$,
then, for every $\theta >  0$,
\begin{align}
\label{eq:Swin-MMOO}
M_{S_{\rm win}}(-\theta,s,  t) \le
\left( (m_+(-\theta)^d
+ d e^{-\theta w}\right)^{\left\lfloor\frac{t-s}{d}
\right\rfloor}\,,
\end{align}
where $m_+(-\theta)$ is the larger eigenvalue
of the matrix $L(-\theta)$ defined by Eq.~\eqref{eq:MMOO-effC-L}.
\end{theorem}

\begin{proof} 
By Lemma~\ref{lem:pos-corr} (Parts 1 and 3), we have
for each choice of $\tau_o,\dots,\tau_n$ in $C_n(s,t)$
the bound
\begin{align*}
E\left[e^{-\theta 
\sum_{i=1}^{n}
\left(S(\tau_{i-1},\tau_i-d)+S(\tau_n , t)\right)}
\right] \le M_S(-\theta, 0 ,t-s-nd) 
\le \left(m(-\theta)\right)^{t-s-nd}\,.
\end{align*}
Inserting this estimate into the
proof of Theorem~\ref{thm:Swin-VBR} gives the result.
\end{proof}

The theorem implies the following lower bound on the effective
capacity of the service process with feedback.

\begin{corollary}\label{cor:EffCap-MMOO}
Under the assumptions of Theorem~\ref{thm:Swin-MMOO},
the effective capacity $\gamma_{\rm win}$ of the
MMOO process $S(s,t)$ with feedback satisfies
for $\theta >0$ 
\begin{align}
\label{eq:effC-MMOO}
\gamma_{\rm win} (-\theta) \ge 
\gamma_S(-\theta) 
-\frac{1}{d\theta}\log \left(
1  + d e^{\theta (d\gamma_S(-\theta)   -w )}\right)\,,
\end{align}
where $\gamma_S(-\theta)$ is given by Eq.~\eqref{eq:MMOO-effC}.
\end{corollary}

Both results extend to more general
two-state Markov-modulated service processes.

\begin{corollary} 
\label{cor:MMP}
Let $S(s,t)$ the two-state Markov-modulated process 
at the end of Subsection IV-A.
If $p_{01}+p_{10}<1$, then the moment-generating
function of the service process with feedback satisfies
Eq.~\eqref{eq:Swin-MMOO}, where
where $m_+(-\theta)$ is the larger eigenvalue
of the matrix $L(-\theta)$ defined by Eq.~\eqref{eq:MMP-effC-L}.
Its effective capacity $\gamma_{\rm win}(-\theta)$
satisfies Eq.~\eqref{eq:effC-MMOO}, where
$\gamma_S(-\theta)=-\frac{1}{\theta} \log m_+(-\theta)$
is the effective capacity of the service process.
\end{corollary}

{\em Remark:} The formula 
$\gamma(-\theta)=-\frac{1}{\theta} \log m_+(-\theta)$
agrees with~\cite[Eq.~(7.17)]{Book-Chang}. 

\begin{proof} In the proof Theorem~\ref{thm:Swin-MMOO},
replace Lemma~\ref{lem:pos-corr} by 
Lemma~\ref{lem:pos-corr-MMP}. 
\end{proof}


An important application of Corollary~\ref{cor:MMP} is 
a leftover service model with 
a VBR server (as analyzed in Sec.~\ref{sec:cbrvbr}), where  
cross-traffic arrivals are governed by a two-state Markov-modulated process. 
Let $c_k$ denote the total service available to all traffic flows in time slot $k$,
and let $a^{\rm c}_k$ denote the cross-traffic arrivals in that time slot. 
Then the leftover service available to the
through flow can be bounded from below by
$$
S^{\rm lo}(s,t)= \left[\sum_{k=s}^{t-1} (c_k-a^{\rm c}_k)\right]^+
\ge \ \sum_{k=s}^{t-1} (c_k-a^{\rm c}_k)\,, 
$$
where we use $[x]^+= \max \{x,0\}$. 
Corollary~\ref{cor:MMP} applies
since the sum on the right-hand side is a two-state Markov-modulated process. It provides a non-trivial lower bound on the service
of the feedback system, provided that the stability
condition $E[a^{\rm c}_k]<E[c_k]$ is met.

\subsection{Quality of the MMOO bounds}

As with the VBR server, we can use the  bounds of
Theorem~\ref{thm:Swin-apriori} to test the accuracy of the 
results on the MMOO process. 
Inserting the parameters of the MMOO process in
Theorem~\ref{thm:Swin-apriori}, $S_{\rm win}$ is bounded by 
\begin{align}
\label{eq:Swin-MMOO-w/d}
\sum_{k=s}^{t-1} \min\left\{P, \frac w d\right\} X_k
\ \le \ S_{\rm win}(s,t) \ \le \ \min \left\{ S(s,t), \left\lceil\frac {t-s} d \right
\rceil w \right\}\,.
\end{align}
The lower bound on the left corresponds to a service process
of a MMOO process with ON rate 
$P'=\min\left\{P,\frac w d \right\}$. The lower bound is exact for 
when $d=1$. 
Its moment-generating
function can  be 
bounded with the help of Lemma~\ref{lem:pos-corr} (Part~3).
For the upper, we can only exploit  
$S_{\rm win}(s,t) \le \lceil\tfrac {t-s} d  \rceil w$, since we do not have a simple upper bound on the
$\eps$-quantiles of the MMOO process.

We also apply Lemma~\ref{lem:pos-corr} (Part 4) to the upper bound.
Using these expressions in Eq.~\eqref{eq:Eff-capacity} 
yields the following bounds for the effective capacity: 
\begin{corollary}\label{cor:apriori-MMOO}
Under the assumptions of Theorems~\ref{thm:Swin-MMOO}, $\gamma_{\rm win} (-\theta)$ 
is bounded for $\theta>0$ by 
\begin{align}
\label{eq:effC-MMOO-w/d}
\gamma'_{\rm win}(-\theta)\ \le\
\gamma_{\rm win}(-\theta) \ \le \ 
\min \left\{
\gamma_S(-\theta), \frac{w}{d}\right\}\,,\quad (\theta>0)\,.
\end{align}
\end{corollary}
Here, $\gamma'_{\rm win}(-\theta)$
is the effective capacity of an MMOO process
with peak rate $P'=\min\{P,\frac w d\}$.
Note that the average rate of this
process, given by $\gamma_{\rm win}'(0)$, lies strictly
below the upper bound, $\min\{\gamma_S(0,\frac{w}{d}\}$.
In the limit $\theta\to\infty$, the three bounds 
in Eqs.~\eqref{eq:effC-MMOO} and~\eqref{eq:effC-MMOO-w/d}
become sharp.


\subsection{Numerical evaluation of MMOO bounds}
We present numerical examples for an MMOO server with 
window flow control, proceeding in a similar fashion as in the 
evaluation of the VBR server. 
The parameters of the MMOO process are selected as 
\[
p_{00} = 0.2, \ \ p_{11} = 0.9, \ \ P = 1.125~\text{Mb} \, .    
\]
With time slots of length 1~ms, 
the server has an average rate of 1~Gbps, which is the same rate as 
that of the VBR service evaluated in Subsec.~\ref{subsec:eval-VBR}. 

\begin{figure}[!t]
\centering
	\subfigure[$w/d= 100$~Mbps.]{
 \includegraphics[width=0.8\textwidth]{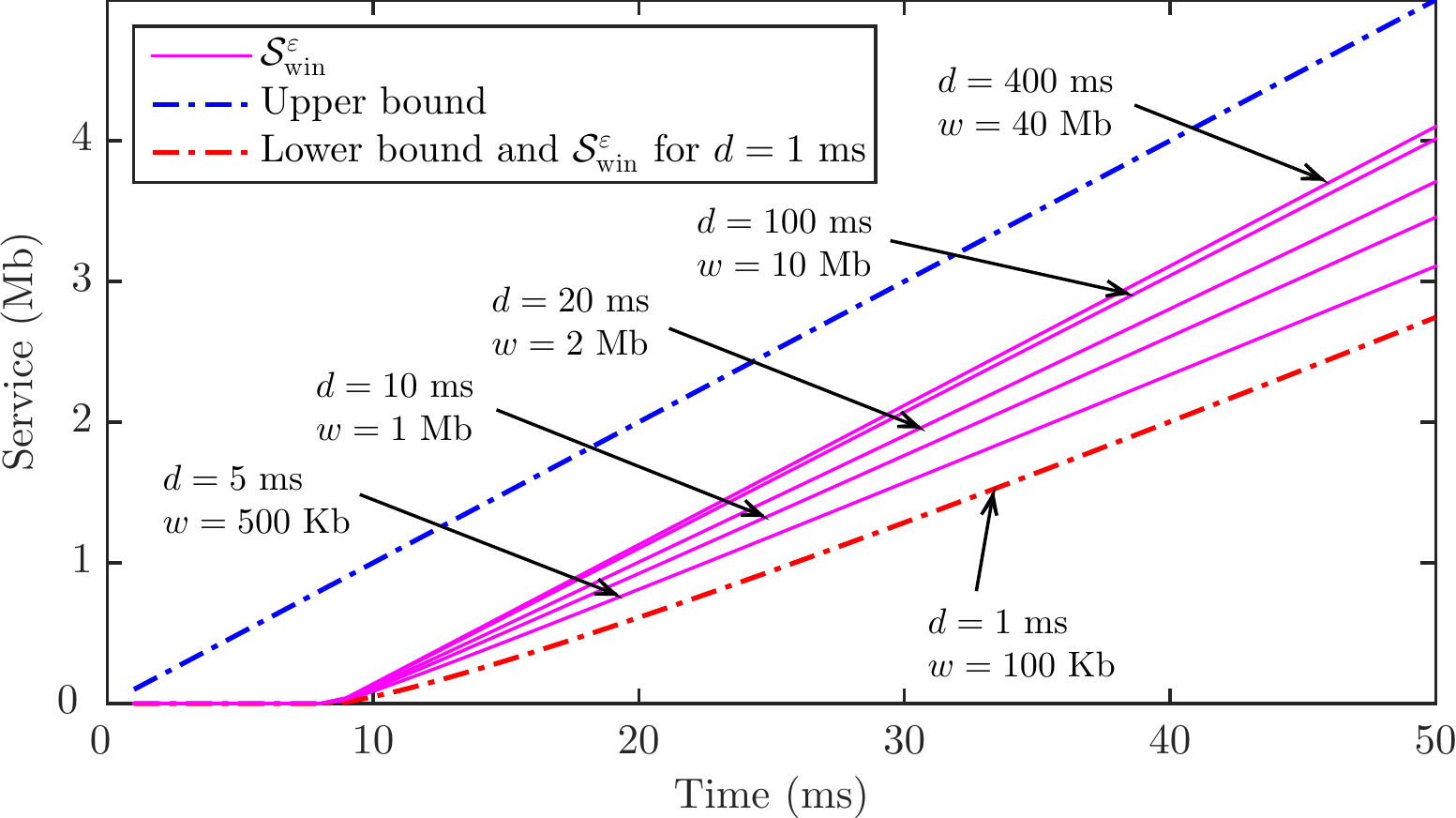}
		\label{fig:MMOO-Swin-100}
	}
	
\centering
	\subfigure[$w/d= 500$~Mbps.]{
		\includegraphics[width=0.8\textwidth]{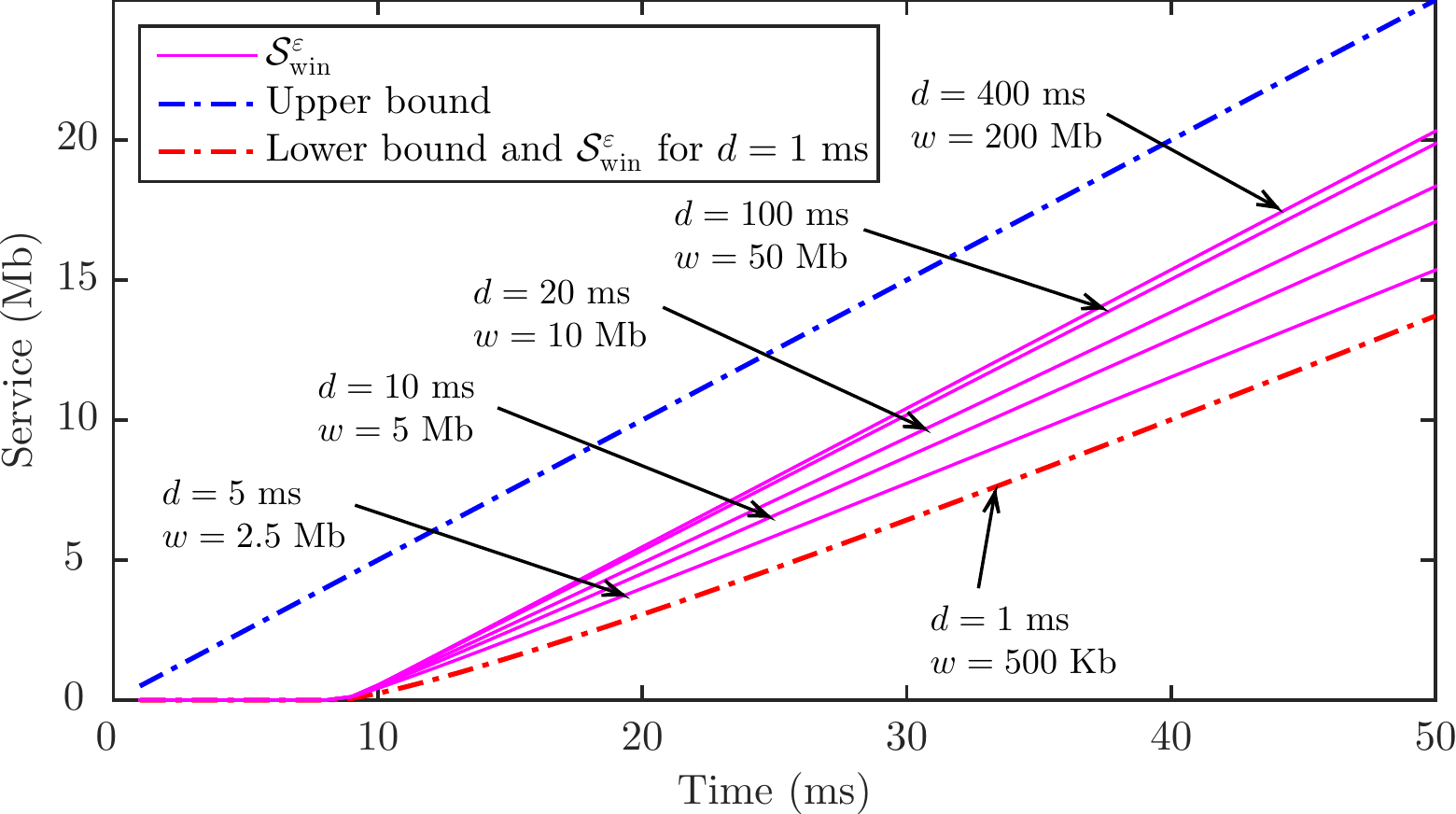}
		\label{fig:MMOO-Swin-500}
	}
\caption{Statistical service curves $\S^\eps_{\rm win} (0,t)$ for time-correlated (Markov-modulated On-Off)  service  ($\eps = 10^{-6}$).}
\label{fig:MMOO-Swin}

\vspace{-17pt}
\end{figure}
\begin{figure}[!t]
\centering
	\subfigure[$w/d= 100$~Mbps.]{
 \includegraphics[width=0.8\textwidth]{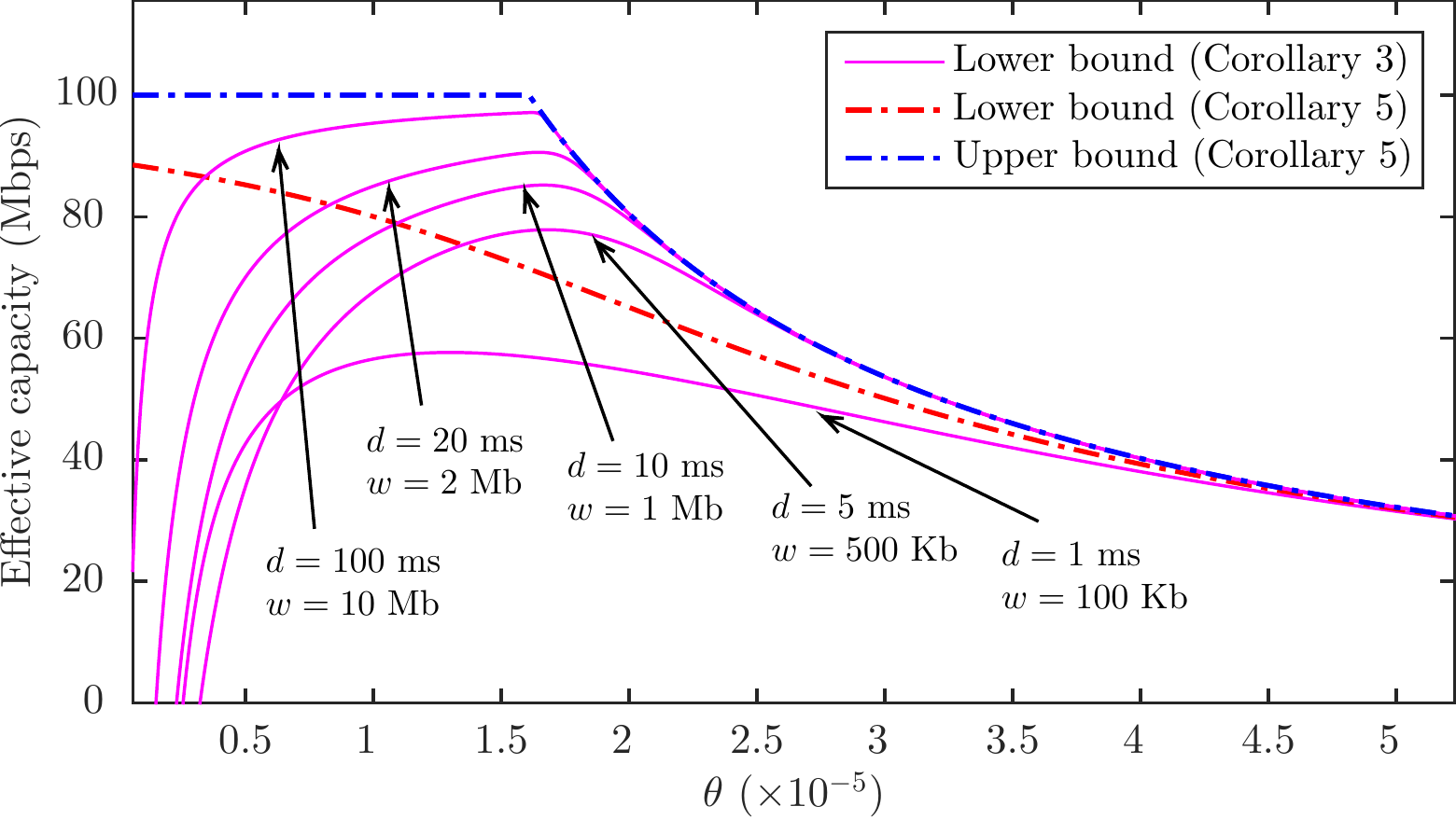}
		\label{fig:MMOO-EffCap-100}
	}
		
\centering
	\subfigure[$w/d= 500$~Mbps.]{
\includegraphics[width=0.8\textwidth]{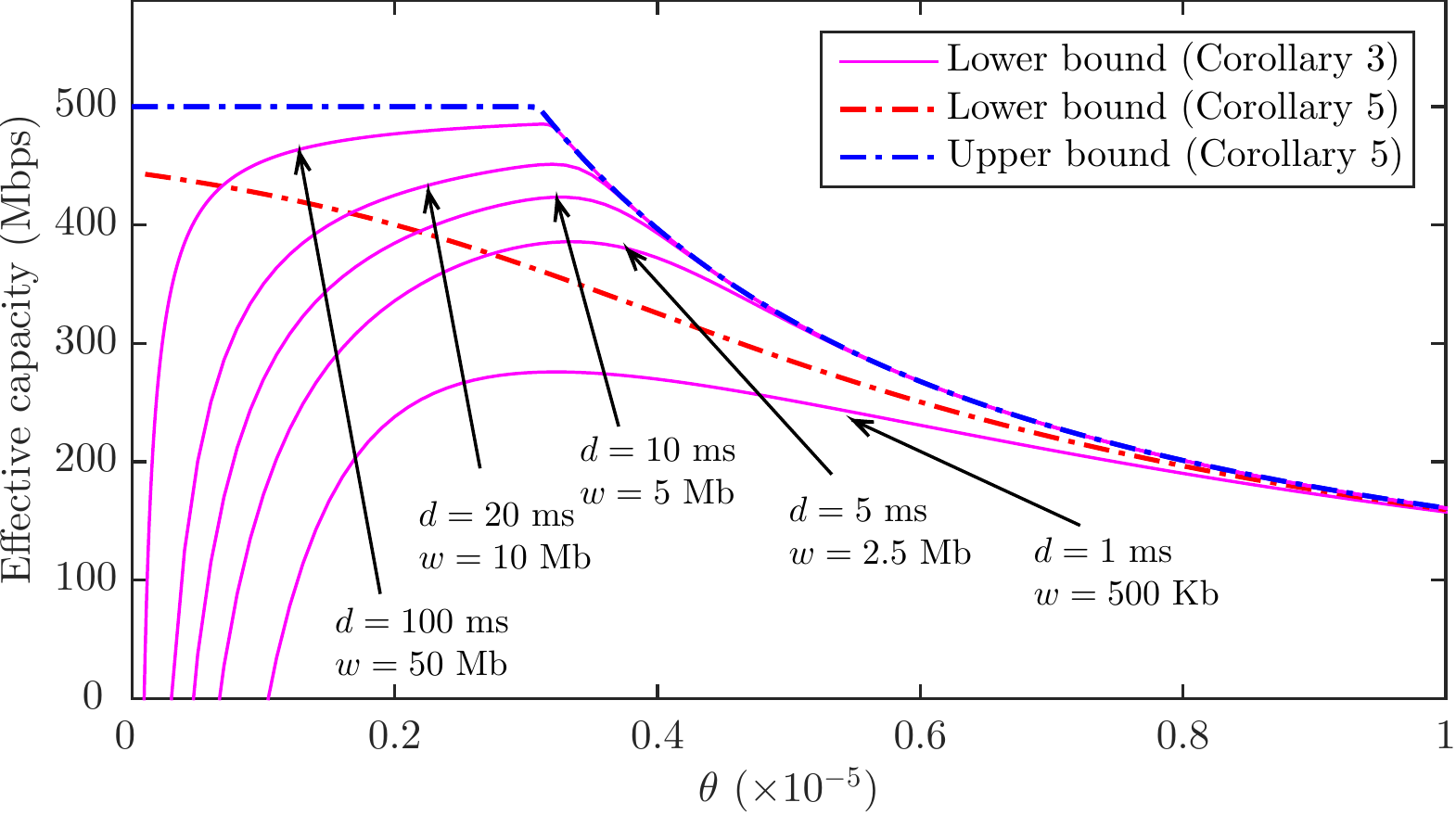}
		\label{fig:MMOO-EffCap-500}
	}			
\caption{Effective capacity $\gamma_{\rm win} (-\theta)$ for  time-correlated (Markov-modulated On-Off) service.}
\label{fig:MMOO-EffCap}

\vspace{-17pt}
\end{figure}

In Fig.~\ref{fig:MMOO-Swin} 
we present the statistical service curve  $\S_{\rm win}^\eps (0,t)$ as a function of time 
where we fix $w/d=100$~Mbps in Fig.~\ref{fig:MMOO-Swin-100} and 
$w/d=500$~Mbps in Fig.~\ref{fig:MMOO-Swin-500}. We set 
$\eps=10^{-6}$.  
The statistical service curves, plotted as solid 
lines, have been constructed with the bound 
on $M_\Swin$ from Theorem~\ref{thm:Swin-MMOO}.
We include for comparison, 
upper and lower bounds obtained from Theorem~\ref{thm:Swin-apriori}
via Eq.~\eqref{eq:Swin-MMOO-w/d} 
represented by dash-dotted lines. 
Note that the lower bounds become positive 
only for $t>10$~ms, and that this is well matched by the  (solid line) statistical service curves.  The initial 
latency is  a property of the MMOO service process, which may 
reside for extended time periods in the OFF state.  
As seen in the VBR service in  Fig.~\ref{fig:VBR-Swin}, 
the statistical service curves increase when  
increasing $w$ and $d$ proportionally. 
Recall that the lower bound is also a statistical service curve when $d=1$~ms. 
As another comparison 
with the VBR service, we note that the upper and lower bounds in 
Fig.~\ref{fig:MMOO-Swin} are separated by a wider margin than 
in Fig.~\ref{fig:VBR-Swin}. This is due to the simpler upper bound, 
since we have not derived an upper bound for the
$\eps$-quantiles of the MMOO process.

Fig.~\ref{fig:MMOO-EffCap} shows the bounds on the effective capacity 
$\gamma_{\rm win} (-\theta)$ as a function of $\theta$, 
where we use $w/d = 100$~Mbps in Fig.~\ref{fig:MMOO-EffCap-100}, 
and $w/d = 500$~Mbps in Fig.~\ref{fig:MMOO-EffCap-500}. 
The bounds from  Corollary~\ref{cor:EffCap-MMOO}, 
for different values of $w$ and $d$, 
are shown as solid lines. 
The upper and lower bounds obtained from
Corollary~\ref{cor:apriori-MMOO} are depicted 
as dash-dotted lines. 
When the bounds from Corollary~\ref{cor:EffCap-MMOO} fall 
below those of Corollary~\ref{cor:apriori-MMOO}, the 
better bound should be used. 
Except for small values of $\theta$, the lower bounds of Corollary~\ref{cor:EffCap-MMOO} 
are close to the upper bound 
from Corollary~\ref{cor:apriori-MMOO}, indicating that the effective capacity can be 
accurately computed.

\section{Conclusions}
\label{sec:conclusions}

We have approached a  well-known open  problem in 
the stochastic network calculus, i.e., an extension of 
the analysis of feedback systems.  
Our analysis addressed a window flow control system with stochastic service, 
where we considered a service with and without time correlations. 
We analyzed and then addressed the difficulty of accounting for the 
time correlations introduced by  feedback mechanisms. 
Our analysis revealed major differences between deterministic and random 
window flow control systems, in particular, an additional dependency 
between the characteristic time scales of the feedback delay and the service process. 
We provided  lower as well as upper bounds on the available service of 
the feedback system, which enabled us to discuss the accuracy of our results, 
and discovered special cases where exact expressions for the service can be obtained. 
The results in this paper can be extended in many directions. Obvious generalizations are to consider random feedback delays and time-variable window sizes. We chose a window flow control network since a corresponding deterministic network calculus analysis exists in the literature. 
There are numerous  
other feedback systems that await an analysis of their backlog and/or delay properties.
Applying our analysis to the detailed dioid algebraic models of TCP feedback in \cite{baccelliTCP} is a logical first step.






\end{document}